\documentclass{article}
\usepackage{amsthm}
\usepackage{amsmath,amssymb}
\usepackage{paralist}

\usepackage{tikz}
\usetikzlibrary{positioning}

\newtheorem{definition}{Definition}

\newtheorem{theorem}{Theorem}
\newtheorem{lemma}{Lemma}

\newcommand{\ignore}[1]{}
\usepackage{color}

\newcommand{\Rantwo}[2]{{\tiny \textcolor{green}{#1}}{\textcolor{red}{#2}}}

\newcommand{\Ran}[1]{\textcolor{red}{#1}}

\newcommand{\hC}{\hat{C}}
\newcommand{\hB}{\hat{B}}

\newcommand{\abs}[1]{| #1 |}
\newcommand{\norm}[1]{\| #1 \|}
\newcommand{\set}[2]{\left\{ #1 \;|\; #2 \right\}}
\newcommand{\sset}[1]{\left\{ #1 \right\}}
\newcommand{\xmax}{x_{\max}}
\newcommand{\xeq}{x_{\mathit{eq}}}

\newcommand{\floor}[1]{\lceil #1 \rceil}
\newcommand{\N}{\mathbb {N}}
\newcommand{\tu}{\tilde{u}}
 \newcommand{\tp}{\tilde{p}} 
\newcommand{\tf}{\tilde{f}}
\newcommand{\tr}{\tilde{r}}

\newcommand{\hp}{\hat{p}}

\newcommand{\hf}{\hat{f}}
\newcommand{\hr}{\hat{r}}
\newcommand{\hx}{\hat{x}}

\newcommand{\barf}{\bar{f}}
 \newcommand{\barp}{\bar{p}} 
\newcommand{\barr}{\bar{r}}

\begin{document}

\title{A Combinatorial Polynomial Algorithm for the Linear Arrow-Debreu Market\footnote{A preliminary version of this paper was presented at ICALP 2013.}}

\author{Ran Duan\thanks{Max-Planck-Institut f\"ur Informatik, Saarbr\"ucken, Germany, {duanran@mpi-inf.mpg.de}}\and Kurt Mehlhorn\thanks{Max-Planck-Institut f\"ur Informatik, Saarbr\"ucken, Germany, {mehlhorn@mpi-inf.mpg.de}}}

\maketitle


\tableofcontents

\begin{abstract}
We present the first combinatorial polynomial time algorithm for
computing the equilibrium of the Arrow-Debreu market model with
linear utilities. Our algorithm views the allocation of money as
flows and iteratively improves the balanced flow as in [Devanur et
al. 2008] for Fisher's model. We develop new methods to carefully
deal with the flows and surpluses during price adjustments. Our algorithm performs $O(n^6 \log (nU))$ maximum flow computations,
where $n$ is the number of agents and $U$ is the maximum integer
utility. The flows have to be presented as numbers of bitlength $O(n\log (nU))$
to guarantee an exact solution. Previously, [Jain 2007, Ye 2007] have given
polynomial time algorithms for this problem, which are based on
solving convex programs using the ellipsoid algorithm and the interior-point method, respectively. 
\end{abstract}

\section{Introduction}\label{sec: introduction}

We provide the first combinatorial polynomial algorithm for
computing the equilibrium of the linear version of an economic market model formulated by Walras\footnote{Walras' model is actually more general and also involves production.} in
1874~\cite{Walras1874}. The model is also known as the linear exchange model. In this model, every agent has an initial
endowment of some goods and a utility function over sets of goods. Goods are assumed to be divisible. 
The market clears at a set of prices if each agent can spend its entire budget (= the total value of its goods at the set of prices) on a bundle of goods with maximum utility, and all goods are completely sold. Market clearing prices are also called equilibrium prices. In the linear version of the problem, all utility functions are linear. 

Formally, the linear model is defined as follows. We may assume w.l.o.g.~that the number of goods is equal to the number of agents and that the $i$-th agent owns the $i$-th good. Let $u_{ij} \ge 0$ be the utility for buyer $i$ if all of good $j$ is allocated to him. If fraction $x_{ij}$ of good $j$ is allocated to buyer $i$, the total utility for  $i$ is 
\[          \sum_{j} u_{ij} x_{ij}. \]
We make the standard assumption that each agent likes some good, i.e., for all $i$, $\max_j u_{ij} > 0$, and that each good is liked by some agent, i.e., for all $j$, $\max_i u_{ij} > 0$.
Agents are selfish and spend money only on goods that give them maximum utility per unit of money, i.e., if $p = (p_1,\ldots,p_n)$ is a price vector then buyer $i$ is only willing to buy goods $j$ with $u_{ij}/p_j = \max_j
  u_{ij}/p_j$. An \emph{equilibrium} is a vector $p$ of positive prices and allocations $x_{ij} \ge 0$ such
  that
\begin{itemize}
\item all goods are completely sold:\quad $\sum_i x_{ij} = 1$
\item all money is spent: \quad $\sum_j x_{ij} p_j  = p_i$
\item only maximum utility per unit of money goods are bought:
\[    x_{ij} > 0 \quad \Rightarrow \quad \frac{u_{ij}}{p_j} = \alpha_i,
\text{ where } \alpha_i = \max_\ell \frac{u_{i\ell}}{p_\ell} \]
\end{itemize}

The existence of an equilibrium is non-obvious. Walras argued existence algorithmically. Assume that a vector $p$ of prices is not an equilibrium vector. Then there is a good for which demand exceeds supply. Increase the price of the good until demand equals supply. Of course, this might destroy the equilibrium for other goods. However, the effect on the other goods will be smaller (he did not define this term) and hence the process will converge. 

Fisher~\cite{FisherThesis} introduced a somewhat simpler model (agents come with budgets) in 1891 and showed how to compute an equilibrium for the case of three agents and three goods by an analog computer. The computer is a hydro-mechanical device that stabilizes at equilibrium prices and allocations~\cite{Brainard-Scarf}. 

The first rigorous existence proof is due to Wald~\cite{Wald36}.
It required fairly strong assumptions. The existence proof for the
general model was given by Arrow and Debreu~\cite{AD1954} in 1954.
\Ran{He} proved that \Ran{the} market equilibrium always exists if
the utility functions are concave. The result is prominently
mentioned in their Nobel prize laudation and the market is usually
referred to as the ``Arrow-Debreu market''. However, their proof
is based on Kakutani's fixed-point theorem and hence
non-constructive. Gale~\cite{Gale57,Gale76} gives necessary and sufficient conditions for the existence of an equilibrium in the linear model; see Section~\ref{sec:general}.

The development of algorithms started in the 60s. There is a wide literature in the economics and mathematics community, and since 2000 also in the algorithm community. The survey paper by Codenotti, Pemmaraju and Varadarajan~\cite{Codenotti-Survey} surveys the literature until 2004.  Early algorithms, for example by Scarf, Smale, Kuhn, and Todd, are inspired by fixed-point proofs or are Newton-based and compute approximations. The running time of these algorithms is exponential. 

First approximation algorithms with polynomial running time appeared after 2000. 
Jain, Madhian, and Saberi~\cite{JMS03} gave a polynomial time approximation scheme. Devanur and Vazirani~\cite{DV03} obtained a strongly polynomial time  approximation scheme and Garg and Kapoor~\cite{Garg-Kapoor} gave a simplified approximation scheme. 
Recently, Ghiyasvand and Orlin~\cite{MJ12} gave an approximation scheme with running time $O(\frac{n}{\epsilon}(m+n\log n))$; here $m$ is the number of positive utilities $u_{ij}$. The approximation algorithms are combinatorial. 

Exact algorithms are also known. Eaves~\cite{Eaves76} presented the first exact algorithm for the linear exchange model. He reduced the model to a linear complementary problem which is then solved by Lemke's algorithm. The algorithm is not polynomial time. Garg, Mehta, Sohoni, and Vazirani~\cite{Garg-Mehta-Sohoni-Vazirani} give a combinatorial interpretation of the algorithm. Polynomial time exact algorithms were obtained based on the 
characterization of equilibria as the solution set of a convex program. The recent paper by Devanur, Garg, and V\'{e}gh~\cite{Devanur-Garg-Vegh} surveys these characterizations. For example, Nenakov and Primak~\cite{Nenakov-Primak} showed that equilibria are precisely the solutions of the system
\[  p_i > 0 \quad x_{ij} \ge 0\quad \sum_j u_{ij} x_{ij} \ge \frac{u_{ik}}{p_k} p_i  \quad\text{for all $i$ and $k$}.\]
Note that $\frac{u_{ik}}{p_k} p_i$ is the utility obtained by the $i$-th buyer if he spends his entire budget, which is $p_i$ under the assumption that his good is completely sold, on good $k$. The program above is not convex in the variables $p_i$ and $x_{ij}$. However, after the substitution $p_i = e^{\pi_i}$ it becomes a convex program in the variables $\pi_i$ and $x_{ij}$. Jain~\cite{Jain07} rediscovered this convex program and showed how to solve it with a nontrivial extension of the ellipsoid method. His algorithm is the first polynomial time algorithm for the linear Arrow-Debreu market. Ye~\cite{Ye2007} showed that polynomial time can also be achieved with the interior point method. The latter algorithm runs in time $O(n^4 \log U)$ for integral utilities bounded by $U$ and is significantly faster than the algorithm presented in this paper. 

In Fisher markets each agent comes with a budget to the market. It is a special case of the Arrow-Debreu market and algorithmically simpler. Eisenberg and Gale~\cite{EG1958} characterized equilibria by a convex program. With the advent of the ellipsoid method, their characterization gave rise to a polynomial algorithm. The first combinatorial algorithm was given by Devanur, Padimitriou, Saberi, and Varzirani~\cite{DPSV08}. Their algorithm uses a maximum flow algorithm as a
black box.  When the input data are integral, their algorithm needs $O(n^5\log U+n^4\log e_{max})$ max-flow
computations, where $n$ is the number of buyers, $U$ the largest
integer utility, and $e_{max}$ the largest initial amount of money of a buyer. If we use the common $O(n^3)$ max-flow algorithm (see~\cite{AMO93}), their running time is $O(n^8\log U+n^7\log e_{max})$. Orlin~\cite{Orlin10} improved the running time to $O(n^4\log U+n^3\log
e_{max})$ and also gave the first strongly polynomial algorithm
with running time $O(n^4\log n)$. Our work is inspired by these papers.

Our algorithm has the advantage of being combinatorial (see Figure~\ref{fig:algo} for a
complete listing), and hence, gives additional
insight in the nature of the problem. We obtain equilibrium prices
by a simple procedure that iteratively adjusts prices and
allocations in a carefully chosen, but intuitive manner. We describe a basic version of the algorithm in Section~\ref{sec:algo}. The basic version already achieves a polynomial number of arithmetic operations on rationals; however, the bitlength of the rationals may be exponential in the size of the problem instance. In Section~\ref{sec: poly time} we achieve polynomial time by the use of approximate arithmetic. The algorithm solves $O(n^6 \log (nU))$ flow problems and needs to deal with numbers of bitlength $O(n \log(nU))$. With an $O(n^3)$ maxflow algorithm, the running time becomes $O(n^{10} \log^2 (nU))$. We introduce some basic concepts in Section~\ref{sec:model} and discuss further related work in Section~\ref{sec:relatedwork}. In Section~\ref{sec: overview}, we give a high-level view of our algorithm. 

The problem of finding a strongly
polynomial algorithm for the linear Arrow-Debreu model remains open.

\subsection{Further Related Work}\label{sec:relatedwork}

There are also algorithms for the Arrow-Debreu model with
non-linear utilities~\cite{CMV05,CMPV05}. The CES (constant
elasticity of substitution) utility functions have drawn much
attention; here, the utility functions are of the form
$u(x_1,...,x_n)=(\sum_{j=1}^nc_jx_j^\rho)^{1/\rho}$ for $-\infty<\rho<1$
and $\rho\neq0$. Codenotti, McCune, Penumatcha, and Varadarajan~\cite{CMPV05} have shown that for $\rho>0$ and $-1\leq\rho<0$, there are polynomial algorithms based on a convex program. In contrast,
Chen, Paparas and Yannakakis~\cite{CPY12} have shown that it is
PPAD-hard to solve market equilibrium of CES utilities for $\rho<-1$.
They also define a new concept ``non-monotone utilities'', and
show the PPAD-hardness to solve the markets with non-monotone
utilities. It remains open to find the exact border between
tractable and intractable utility functions.

All algorithms mentioned so far are centralized in the sense that the algorithms need to know all the utilities at the start of the algorithm and, in the case of iterative algorithms, the global state of the market, i.e., all prices and the demand for each good. Cole and Fleischer~\cite{Cole-Fleischer} and Cheung, Cole, and Rastogi~\cite{Cheung-Cole-Rastogi} explore local algorithms, where each agent only knows the price and demand of his good and, moreover, the market is run for many periods.

\subsection{Model and Definitions}\label{sec:model}

The linear exchange model is defined by the following assumptions; they are as in Jain's paper~\cite{Jain07}:\medskip

\begin{compactenum}[1.]
    \item There are $n$ agents in the system. Each agent $i$ has only one good, which is different from the goods other people have. Agent $i$ owns good $i$.
    \item There is one unit of each good $i$. So, if the price of good $i$ is $p_i$, agent $i$ will obtain $p_i$ units of money when selling its good completely. 
    \item Each agent $i$ has a linear utility function $\sum_j u_{ij}x_{ij}$, where $x_{ij}$ is the amount of good $j$ consumed by $i$.
    \item For all $i$, there is a $j$ such that $u_{ij}>0$. (Everybody likes some good.)
    \item For all $j$, there is an $i$ such that $u_{ij}>0$. (Every good is liked by somebody.)
    \item (Irreducibility) For every proper subset $P$ of agents, there exist
    $i\in P$ and $j\notin P$ such that $u_{ij}>0$.
  \item Each $u_{ij}$ is an integer between 0 and $U$.
\end{compactenum}\medskip

Assumptions 1 to 5 are standard. Assumption 6 simplifies the presentation, and we will come back to it in Section~\ref{sec:general}. Assumptions 1 to 6 guarantee that an equilibrium with positive prices exist~\cite{Gale76}. Assumption 7 is for the purposes of polynomial time computation. 

Let $p=(p_1,p_2,...,p_n)$ denote the vector of prices of goods $1$
to $n$, so they are also the budgets of agents $1$ to $n$, if goods are completely sold. 
In this paper, it is useful to represent each agent $i$ twice, once in its role as a buyer and once in his role as (the owner of) a good. We denote the set of  buyers by
$B=\{b_1,b_2,...,b_n\}$ and the set of goods by
$C=\{c_1,c_2,...,c_n\}$. So, if the price of goods $c_i$ is $p_i$,
buyer $b_i$ will have $p_i$ amount of money. For a subset $B'$ of
agents or a subset $C'$ of goods, we also use $p(B')$ or $p(C')$
to denote the total prices of the goods the agents in $B'$ own or of
the goods in $C'$. For a vector $v=(v_1,v_2,...,v_k)$, let:\medskip
\begin{compactenum}[\hspace{\parindent}--]
    \item $|v|=|v_1|+|v_2|+...+|v_k|$ be the $l_1$-norm of $v$.\smallskip
    \item $\|v\|=\sqrt{v^2_1+v^2_2+...+v^2_k}$ be the $l_2$-norm of $v$.
\end{compactenum}\medskip
Each agent only buys its favorite goods, that is, the
goods with the maximum ratio of utility and price. Define the \emph{bang per buck} of buyer $b_i$ to be
$\alpha_i=\max_j\{u_{ij}/p_j\}$. 

For a price vector $p$, the \emph{equality network} $N_p$ is a flow network with vertex set
$\sset{s,t} \cup B \cup C$, where $s$
is a source node, $t$ is a sink node, $B$ is the set of buyers, and $C$ is the set of goods, and the following edge set: \medskip

\begin{compactenum}[\hspace{\parindent}--]
    \item An edge $(s,b_i)$ with capacity $p_i$ for each $b_i \in B$. 
    \item An edge $(c_i,t)$ with capacity $p_i$ for each $c_i \in C$.
    \item An edge $(b_i,c_j)$ with infinite capacity whenever $u_{ij}/p_j=\alpha_i$. We use $E_p$ to denote these edges. 
\end{compactenum}\medskip

Our task is to find a price vector $p$ such that there is a
flow in which all edges from $s$ and to $t$ are saturated, i.e.,
$(s,B\cup C\cup t)$ and $(s\cup B\cup C, t)$ are both minimum
cuts. When this is satisfied, all goods are sold and all of the
money earned by each agent is spent on goods of maximum utility per unit of money.

For a set $B'$ of buyers define its neighborhood $\Gamma(B')$ in $E_p$ to be:
$\Gamma(B')=\set{c \in C}{\text{$(b,c)\in
E_p$ for some $b \in B'$}}$. 
Clearly, there is no edge in $E_p$ from $B'$ to
$C\setminus\Gamma(B')$.

With respect to a flow $f$, define the surplus $r(b_i)$ of a buyer $i$ to be
the residual capacity of the edge $(s,b_i)$, and define the
surplus $r(c_j)$ of a good $j$ to be the residual capacity of the
edge $(c_j,t)$. That is, $r(b_i)=p_i-\sum_j f_{ij}$, and
$r(c_j)=p_j-\sum_i f_{ij}$, where $f_{ij}$ is the amount of flow
in the edge $(b_i,c_j)$. Define the surplus vector of buyers to be
$r(B)=(r(b_1),r(b_2),...,r(b_n))$. Also, define the total surplus
to be $|r(B)|=\sum_i r(b_i)$, which is also $\sum_j r(c_j)$ since
the total capacity from $s$ and to $t$ are both equal to $\sum_i
p_i$. For convenience, we denote the surplus vector of flow $f'$
by $r'(B)$. In the network corresponding to market clearing
prices, the total surplus of a maximum flow is zero.

\subsection{Overview of our algorithm}\label{sec: overview}

The overall structure of our algorithm is shown in Figure~\ref{fig:high-level-view} and is similar to the
ones of Devanur et al.~\cite{DPSV08} and
Orlin~\cite{Orlin10} for computing equilibrium prices in Fisher
markets. However, the details are quite different. The algorithm
works iteratively. It starts with all prices $p_i$ equal to one and with a balanced flow $f$ in $N_p$. The concept of a balanced flow was introduced in~\cite{DPSV08}. A balanced flow is a flow that minimizes the $l_2$-norm of the surplus vector $r(B)$. Every balanced flow is also a maximum flow. In each iteration, we first analyze the current balanced flow and then carefully adjust prices and flow. The updated flow is not necessarily a balanced flow in the equality network with respect to the new prices. We therefore extend it to a balanced flow. Once the total surplus of the current flow is tiny, we exit from the loop and round the current price vector to a vector of equilibrium prices.

\begin{figure}[t]
\centerline{\framebox{\parbox{0.9\textwidth}{
\begin{tabbing}
555\=555\=555\=555\=\kill
\>Set all prices $p_i$ to one and compute a balanced flow $f$ in $N_p$;\quad\mbox{}\\ \medskip
\>{\bf Repeat}\\
\>\>Adjust prices and flow and obtain a price vector $p'$ and a flow $f'$;\\
\>\>$p = p'$ and $f = f'$;\\
\>\>extend $f$ to a balanced flow in $N_p$;\\
\>{\bf Until} the total surplus of $f$ is tiny; \\ \medskip
\>Round the price vector $p$ to get an exact solution;
\end{tabbing}}}}
\caption{A high-level view of the algorithm.}\label{fig:high-level-view}
\end{figure}

We next give more details about the price and flow update. We number the buyers in order of decreasing surpluses:
$b_1, \ldots, b_n$. We find the minimal $\ell$ such that $r(b_\ell)$ is
substantially larger (by a factor of $1 + 1/n$) than $r(b_{\ell+1})$;
$\ell = n$ if there is no such $\ell$. Let $B' = \{b_1,\ldots,b_\ell\}$ and
let $\Gamma(B')$ be the goods that are adjacent to a buyer in $B'$
in the equality graph. Observe that all flow from buyers in $B'$ goes to goods in $\Gamma(B')$ and all buyers in $B'$ have surplus. Thus the goods in $\Gamma(B')$ have no surplus and the demand for them at the current prices exceeds their supply. Also, there is no flow from the buyers in $B \setminus B'$ to the goods in $\Gamma(B')$; this is due to the fact that the flow is balanced. We raise the prices of the goods in $\Gamma(B')$ by a common factor $x$ and obtain a new price vector $p'$. We also construct a flow $f'$ from $f$ by multiplying the flow on all edges incident to goods in $\Gamma(B')$ by the same factor $x$ and increasing the flow from $s$ to buyers in $B'$ such that flow conservation holds.\footnote{In~\cite{DPSV08,Orlin10} only prices are raised and
flows stay the same. This works for Fisher's model because budgets
are fixed. However, in the Arrow-Debreu model, an increase of
prices of goods will also increase the budgets of their owners. } Observe that the surpluses of the goods in $\Gamma(B')$ stay zero. 

The change of prices and flows affects the surpluses of the buyers, some go up and some go down. More precisely, there are four kind of buyers, depending on whether a buyer $b$ belongs to $B'$ or not and whether the good
owned by $b$ belongs to $\Gamma(B')$ or not. For the definition of the factor $x$, we perform the following thought experiment. We increase the prices of the goods in $\Gamma(B')$ and the flow on the edges incident to them continuously by a common factor $x$ until one of three events happens: (1) a new edge enters
the equality graph
(2) the surplus of a buyer in $B'$ and a buyer in $B \setminus B'$
becomes equal, or (3) $x$ reaches a substantial value ($\xmax = 1 +
1/O(n^3)$ in our algorithm).\footnote{The increase of prices of goods in
$\Gamma(B')$ makes the
  goods in $C \setminus \Gamma(B')$ more attractive and hence an equality
  edge connecting a buyer in $B'$ with a good in  $C \setminus \Gamma(B')$
  may come into existence. This event also exists in~\cite{DPSV08,Orlin10}. Events (2) and (3) have no
parallel in~\cite{DPSV08,Orlin10}.} 

If a new equality edge arises, say $(b,c)$ with $b \in B'$ and $c \not\in \Gamma(B')$, we modify $f'$ further. We use the new equality edge to decrease the surplus of $c$ and to balance the surplus of $b$ with the surplus of the other neighbors of $c$.  

The new flow $f'$ is not necessarily a balanced flow in $N_{p'}$. We therefore extend $f'$ to a balanced flow. We ensure that goods with surplus zero keep surplus zero in this process. This ends the description of the main loop of the algorithm. 

In what sense are we making progress? We use two quantities to measure progress. The first quantity is the maximum price of any good. Prices never decrease, prices are bounded by $(nU)^n$ (Lemma~\ref{thm:largest}), and in every iteration with $x = \xmax$ at least one price increases by the factor $\xmax$. It follows that the number of iterations with $x = \xmax$ is bounded by a polynomial in $n$ and $\log U$ (Lemma~\ref{lem:number-of-bad}); we refer to these iterations as $\xmax$-iterations. 

Our second measure is the $l_2$ norm of the surplus vector of the balanced flow $f$.\footnote{\cite{DPSV08} also uses the $l_2$ norm of the surplus vector as a potential function. They show that it is strictly decreasing from one iteration to the next.} In $\xmax$-iterations, the norm of the surplus vector may go up by a factor $1 + O(1/n^3)$. In the iterations, where either event (1) or (2) occurs, the norm
of the surplus vector decreases substantially, since surplus moves
from a buyer in $B'$ to a buyer in $B \setminus B'$ and buyers
in these two groups have, by the choice of groups, substantially
 different surpluses. We show that the norm decreases by a factor of at least $1 + \Omega(1/n^3)$ (Lemma~\ref{thm:main}). We refer to the latter iterations as balancing iterations. We exit from the main loop, when the total surplus is tiny (less than $\epsilon$). 

Since we know already a polynomial bound on the number of $\xmax$-iterations, we also get a polynomial bound on the number of balancing iterations (Lemma~\ref{lem:number-of-good}). 

The algorithm, as outlined so far, has polynomial arithmetic complexity. We doubt that it has polynomial running time. The computation stays within the rational numbers. However, a naive analysis gives only an exponential bound on the bitlength of prices and flows. In order to guarantee polynomial running time, we show that is suffices to carry out the computation on rationals of polynomial bitlength. Some changes are required to the basic algorithm to make this work, see Section~\ref{sec: poly time}.

\section{The Algorithm}\label{sec:algo} \label{sec: basic algorithm}

At a high level, our algorithm is as shown in Figure~\ref{fig:high-level-view}. We introduce the notion of a balanced flow in Section~\ref{sec:balanced-flow} and our approach to changing prices and flows in Section~\ref{sec:price-adjustments}. We will then give a detailed description of the algorithm in Section~\ref{sec:whole-algorithm}. The extraction of the vector of equilibrium prices is discussed in Section~\ref{sec: rounding after termination}.

\subsection{Balanced flow}\label{sec:balanced-flow}

As in~\cite{DPSV08}, we define the concept of a balanced flow to be
a maximum flow that balances the surpluses of buyers. Balanced flows are not necessarily unique; however, the surplus vector of the buyers is unique.

\begin{definition}[{\rm \cite{DPSV08}}] A balanced flow wrt.~to a price vector $p$ is a maximum flow in the network $N_p$ that minimizes $\|r(B)\|$. 
\end{definition}
For flows $f$ and $f'$ and their surplus vectors $r(B)$ and
$r'(B)$, respectively, if $\|r(B)\|<\|r'(B)\|$, then we say that $f$ is
\emph{more balanced} than $f'$. The next lemma justifies the adjective ``balanced''.

\begin{lemma}[{\rm \cite{DPSV08}}]\label{thm:l2norm}
If $a\geq b_i\geq 0, i=1,2,...,k$ and
$\delta\geq\sum_{i=1}^k\delta_i$, where $\delta,\delta_i\geq 0$,
$i=1,2,...,k$, then:
\[
\|(a,b_1,b_2,...,b_k)\|^2\leq\|(a+\delta,b_1-\delta_1,b_2-\delta_2,...,b_k-\delta_k)\|^2-\delta^2.
\]
\end{lemma}
\begin{proof}
\begin{align*}
(a+\delta)^2+\sum_{i=1}^k (b_i-\delta_i)^2-a^2-\sum_{i=1}^k
b_i^2
&\geq 2a\delta+\delta^2-2\sum_{i=1}^k b_i\delta_i\\
&\geq  \delta^2+2a(\delta-\sum_{i=1}^k\delta_i) \geq \delta^2.
\end{align*}
\end{proof}

Devanur et al~\cite{DPSV08} showed how to compute a balanced flow with $O(n)$ max-flow computations. We need a slight extension of their result which also establishes that the edge-flows in a balanced flow are rational numbers with small denominator, assuming that the input data is integral.

\begin{lemma}[{\rm \cite{DPSV08}}]\label{thm:compute-balanced}
Let $p$ be a price vector.  Given a flow
$f$ in the network $N_p$, a balanced flow $f''$ can be computed with at most $n$ max-flow
computations. Goods with surplus zero wrt.~$f$ also have surplus zero wrt.~$f''$.

If all prices, utilities, and entries of $f$ are integral, the entries of $f''$ are rational numbers with a common denominator $n!$.
\end{lemma}
\begin{proof}
We first turn $f$ into a maximum flow by finding the maximum flow in the residual graph $G_f$ from $s$ to $t$, which takes about $O(n^3)$ time.
Observe that goods with surplus zero keep surplus zero.

If $f$ saturates all edges out of $s$, there is no surplus with respect to $f$ and hence $f$ is balanced. So assume that at least one of the edges out of $s$ is not saturated.
In the residual graph\footnote{The residual graph $G_f$ with respect to a flow $f$ is defined in the standard way. It has vertex set $B \cup C \cup \sset{s,t}$ and edge set $(\sset{s} \times B) \cup E_p \cup E_p^{-1} \cup (C \times \sset{t})$. Edge $(s,b_i)$ has capacity $p_i - f_{si}$, edge $(b_i,c_j) \in E_p$ has infinite capacity, edge $(c_j,b_i) \in E_p^{-1}$ has capacity $f_{ij}$ and edge $(c_j,t)$ has capacity $p_j - f_{jt}$.}  $G_f$ w.r.t. to $f$, let $S\subseteq B\cup
C$ be the set of nodes reachable from $s$, and let $T=(B\cup
C)\setminus S$ be the remaining nodes. There must be at least one unsaturated edge into $t$ and the source of any such edge belongs to $T$. Thus $T$ is nonempty.

There are no edges
from $S\cap B$ to $T\cap C$ in the equality graph since any such edge has infinite capacity and hence would belong to the residual graph, and there is no
flow from $T\cap B$ to $S\cap C$ as the reversals of these edges would belong to the residual graph.
The edges from $s$ to $T \cap B$ and the edges from $S \cap C$ to $t$ are saturated and
$(s \cup S, T \cup t)$ is a minimum cut of the residual graph. Thus, the
buyers in $T\cap B$ and the
goods in $S\cap C$ have no surplus w.r.t.~$f$, and this holds true
for every maximum flow since every maximum flow must saturate all edges in a minimum cut.
Let $G'$ be the network spanned by $s\cup
S\cup t$, and let $f'$ be the balanced maximum flow in $G'$. The flow
$f'$ can be computed by $n$ max-flow computations; Corollary 8.8
in~\cite{DPSV08} is applicable since $(s\cup S,t)$ is a min-cut in
$G'$. Finally, $f'$ together with the restriction of $f$ to
$s\cup T\cup t$ is a balanced flow $f''$ in $G$.

The surpluses of the goods in $S \cap C$ stay at zero and hence unchanged. The surpluses of the goods in $T \cap C$ are clearly unchanged. 

\newcommand{\ravg}{r_{{\mathit avg}}}

For the claim about the rational representation of the entries of $f''$, we have to give more details of the balanced flow algorithm in~\cite{DPSV08}. They use a divide and conquer approach. Consider the residual network $G_f$ with respect to $f$ restricted to the edges connecting vertices in $S$ and let $\ravg$ be the average surplus of the buyers in $S$; $\ravg$ is a rational number with denominator $\abs{S}$. Give every buyer $b_i$ in $S$ a supply/demand of $r(b_i) - \ravg$; positive values are supplies and nonpositive values are demands.\footnote{The nodes $b_i$ with $r(b_i) = \ravg$ may be treated as supply or demand nodes.}. Compute a maximum flow $g$. If all supplies can be routed, there is a balanced flow, namely $f + g$,  in which all buyers in $S$ have surplus equal to $\ravg$. We come to the case that not all supplies can be routed. Let $S_S$ be the nodes in $S$ reachable from the buyers $b_i$ with $r(b_i) > \ravg$ in the residual graph with respect to $f + g$, and let $T_S$ be the remaining nodes. Then $(S_S,T_S)$ is a min-cut and we can recurse on the buyers in $S_S$ and the buyers in $T_S$. Since the cardinalities of both sets of buyers is less than $\abs{S}$, the flows obtained are rational numbers with common denominator at most $n!$. 
\end{proof}

\subsection{Price and Flow Adjustment}\label{sec:price-adjustments}

We introduce a method for price adjustments. We never decrease any price and we
only increase prices of goods with zero surplus that can be reached from buyers with positive surplus. For such goods, demand exceeds supply and hence intuition suggests to make such goods more expensive. Increasing the price of a good and keeping the flow into the good constant, will create surplus at the good. We avoid this and hence whenever we increase the price of a good, we also increase the inflow of the good by the same amount. For goods whose price stays unchanged, the inflow is not changed. In this way, the surplus of any good whose price is increased stays zero. Changing the inflow into some goods also changes the outflow of some buyers. We adjust their inflows accordingly.

How should prices be increased? Let $\Gamma$ be the set of goods for which we increase prices and inflows; we will discuss the choice of $\Gamma$ below. We increase the prices of the goods in $\Gamma$ and the inflow into these goods gradually by a common multiplicative factor $x > 1$; the choice of $x$ will be discussed below. As a consequence,
any ratio $u_{ij}/p_j$ becomes $u_{ij}/(x p_j)$ for $c_j \in \Gamma$ and stays constant for $c_j \not\in \Gamma$. The change may create and destroy equality edges. We stop the increase of $x$ at the latest when an equality edge carrying positive flow will be destroyed by a further increase. We may destroy equality edges carrying no flow, but this does no harm.

For which goods should we increase prices? In a balanced flow
$f$ with surplus vector $r$, we choose a positive surplus bound $S$ and consider the set of buyers $B(S)$ with surplus at least $S$, i.e.,
\[      B(S)=\set{b_i\in B}{r(b_i)\geq S}.\]
Let
\[ \Gamma(B(S)) = \set{c_j \in C}{ (b_i,c_j)\in E_p \text{ for some } b_i \in B(S)} \]
be the set of goods connected to the buyers in $B(S)$ by edges in the equality graph with respect to the current price vector $p$. We increase the prices of the goods in $\Gamma(B(S))$.

\begin{lemma}\label{thm:bal1}
In a balanced flow $f$ in $N_p$, there is no
edge in $E_p$ that carries flow from $B\setminus B(S)$ to $\Gamma(B(S))$. Moreover, the goods in $\Gamma(B(S))$ have no surplus.
\end{lemma}
\begin{proof}
Suppose there is an edge $(b_i,c_j)$ that carries flow such
that $b_i\notin B(S)$ and $c_j\in\Gamma(B(S))$. Then, in the
residual graph, there are directed edges $(b_k,c_j)$ and
$(c_j,b_i)$ with nonzero capacities with $b_k\in B(S)$.
However, $r(b_k)\geq S>r(b_i)$, and hence we can augment along this path
and get a more balanced flow by Lemma~\ref{thm:l2norm}, contradicting that $f$ is already a
balanced flow.

The goods in $\Gamma(B(S))$ have no surplus, since the buyers in $B(S)$ have positive surplus, $f$ is a maximum flow, and all the flow coming form buyers in $B(S)$ goes to the goods in $\Gamma(B(S))$.
\end{proof}

We increase the prices of the goods in $\Gamma(B(S))$ and the flow on the edges incident to these goods by a common multiplicative factor $x > 1$. We will leave the prices of the goods in $C \setminus \Gamma(B(S))$ and the flows on the edges incident to these goods  unchanged. For a buyer $b_i$, we scale the flow on the edge $(s,b_i)$ by a factor $x$ if $b_i \in B(S)$ and leave it unchanged otherwise. Note that for a buyer $b_i \in B(S)$ all outgoing flow goes to goods in $\Gamma(B(S))$, and for a buyer $b_i \not\in B(S)$ none of the outgoing flow goes to a good in $\Gamma(B(S))$.
\begin{align}
\label{new prices}
p'_j &= \begin{cases}
x\cdot p_j     & \text{if $c_j\in \Gamma(B(S))$};\\
p_j & \text{if $c_j\notin\Gamma(B(S))$}.
\end{cases}\\
\label{new flow interior}
f'_{ij} &= \begin{cases}
x\cdot f_{ij}      & \text{if $c_j\in \Gamma(B(S))$};\\
f_{ij}  & \text{if $c_j\notin\Gamma(B(S))$}.
\end{cases}\\
f'_{jt} &= \begin{cases}
x\cdot f_{jt}      & \text{if $c_j\in \Gamma(B(S))$};\\
f_{jt}  & \text{if $c_j\notin\Gamma(B(S))$}.
\end{cases}\\
\label{new flow s}
f'_{si} &= \begin{cases} x \cdot f_{si} & \text{if $b_i \in B(S)$}\\
                                                  f_{si} & \text{if $b_i \not\in B(S)$}\\
\end{cases}
\end{align}
Since there are no edges from $B(S)$ to $C \setminus \Gamma(B(S))$, and the edges from $B \setminus B(S)$ to $\Gamma(B(S))$ are not carrying flow, an equivalent definition of $f'_{ij}$ is
$f'_{ij} = x f_{ij}$ if $b_i \in B(S)$ and $f'_{ij} = f_{ij}$ if $b_i \not\in B(S)$.

We next discuss constraints on $x$. We must make sure that equality edges carrying positive flow stay equality edges and that surpluses of buyers stay nonnegative.

Let us first discuss the effect of the price change on the equality graph.
Equality edges from $B \setminus B(S)$ to $\Gamma(B(S))$ may disappear. Since they are not carrying flow, this is of no concern.  If there are edges $(b_i,c_j)$ and $(b_i,c_k)$ in $E_p$
where $b_i\in B(S), c_j,c_k\in\Gamma(B(S))$, then
$u_{ij}/p_j=u_{ik}/p_k$. Since the prices in $\Gamma(B(S))$ are
multiplied by a common factor $x$, $u_{ij}/p_j$ and $u_{ik}/p_k$
remain equal after a price adjustment. Similarly, if $b_i \not\in B(S)$ and $c_j,c_k \not\in \Gamma(B(S))$, the
ratios $u_{ij}/p_j$ and $u_{ik}/p_k$ are unchanged.
However, the goods in
$C\setminus\Gamma(B(S))$ will become more attractive, so there may
be edges from $B(S)$ to $C\setminus\Gamma(B(S))$ entering the equality
network, and the increase of prices needs to stop when this
happens. Define such a factor to be $\xeq(S)$, that is,
\[
\xeq(S)=\min \set{\frac{u_{ij}}{p_j}\cdot\frac{p_k}{u_{ik}}}{b_i\in B(S),
(b_i,c_j)\in E_p, c_k\notin\Gamma(B(S))}.\]
We need $O(n^2)$ multiplications/divisions to compute $\xeq(S)$.
When we increase the prices of the goods in $\Gamma(B(S))$ by a
common factor $x\leq \xeq(S)$, the equality edges in
$(B(S)\times \Gamma(B(S))) \cup ((B \setminus B(S)) \times (C \setminus \Gamma(B(S))))$ will remain in the network.

We also need to make sure that the surpluses stay nonnegative. The surpluses of goods stay the same; they stay zero for the goods in $\Gamma(B(S))$ since inflow and outflow is changed by the same factor $x$ and they stay unchanged for the goods outside $\Gamma(B(S))$ since inflow and outflow stay the same. The capacities of the edges $(s,b_j)$ and $(c_j,t)$ for $c_j \in \Gamma(B(S))$ are multiplied by $x$, the capacities of the other edges incident to $s$ or $t$ are unchanged. For buyers $b_i \in B(S)$ with $c_i \in C \setminus \Gamma(B(S))$ the surplus goes down as the capacity of the inedge stays at $p_i$ and the outflow $p_i - r(b_i)$ is multiplied by $x$. Thus we need $x \le p_i/(p_i - r(b_i))$ for all such $i$.

We distinguish four kinds of buyers:
\[ \begin{array}{ll}
\text{type 1} & \mbox{if $b_i\in B(S)$ and $c_i\in\Gamma(B(S))$};  \\
\text{type 2} & \mbox{if $b_i\in B(S)$ and $c_i\notin\Gamma(B(S))$}; \\
\text{type 3} & \mbox{if $b_i\notin B(S)$ and $c_i\in\Gamma(B(S))$};\\
\text{type 4} & \mbox{if $b_i\notin B(S)$ and $c_i\notin\Gamma(B(S))$}. 
\end{array}\]

\begin{theorem}\label{thm:adjusted-flow}
Given a balanced flow $f$ in $N_p$, a positive surplus bound $S$, and a parameter $x > 1$ such that  $x\leq\min \set{p_i/(p_i-r(b_i))}{b_i\in B(S), c_i\notin \Gamma(B(S)) }$ and $x\leq \xeq(S)$, the flow $f'$ defined in (\ref{new flow interior}) to (\ref{new flow s}) is a feasible flow in the equality network with respect to the prices in (\ref{new prices}). The surplus of each good remains unchanged, and the
surpluses of the buyers become:
\[
r'(b_i)=\left\{
\begin{array}{ll}
x\cdot r(b_i)  & \mbox{if $b_i\in B(S), c_i\in\Gamma(B(S))$ (type 1)};  \\
(1-x)p_i+x\cdot r(b_i)     & \mbox{if $b_i\in B(S), c_i\notin\Gamma(B(S))$ (type 2)}; \\
(x-1)p_i+r(b_i)        & \mbox{if $b_i\notin B(S), c_i\in\Gamma(B(S))$ (type 3)};\\
r(b_i)          & \mbox{if $b_i\notin B(S),
c_i\notin\Gamma(B(S))$ (type 4)}. 
\end{array}
\right.
\]
\end{theorem}
\begin{proof}
Since the flows on all edges associated with goods in
$\Gamma(B(S))$ are multiplied by $x$, the surplus of each good in
$\Gamma(B(S))$ remains zero. The surplus of type 2 buyers
decreases because the flows from a type 2 buyer $b_i$ are
multiplied by $x$, but its budget $p_i$ is not changed. The flow
after adjustment is $x(p_i-r(b_i))$; this is at most $p_i$.
The new surplus with respect to $f'$ is $r'(b_i)=(1-x)p_i+xr(b_i)$.

Since both money and flows are multiplied by $x$ for a type 1
buyer, their surplus is also multiplied by $x$. For a type 3 buyer
$b_i$, their flows are not changed, but their money is multiplied by
$x$, so the new surplus is $xp_i-(p_i-r(b_i))$. For type 4 buyers, neither flow nor budget changes, and hence the surplus does not change. 
\end{proof}

The surplus of type 1 and 3 buyers increases, the surplus of type 2 buyers decreases, and the surplus of type 4 buyers does not change. We define quantities $x_{23}(S)$ and $x_{24}(S)$ at which the surplus of a type 2 and type 3 buyer, respectively type 4 buyer becomes equal:
\begin{align*}
x_{23}(S) &= \min \set{\frac{p_i + p_j - r(b_j)}{p_i + p_j - r(b_i)}}{\text{$b_i$ is type 2 and $b_j$ is type 3 buyer}},\\
x_{24}(S) &= \min \set{\frac{p_i - r(b_j)}{p_i - r(b_{\Ran{i}})}}{\text{$b_i$ is type 2 and $b_j$ is type 4 buyer}}.
\end{align*}

\subsection{Whole procedure}\label{sec:whole-algorithm}

The whole algorithm is shown in Figure~\ref{fig:algo}. We call one
execution of the loop body an iteration. We use constants $R$, $\epsilon$, $\Delta$, and $\xmax$ as defined in the first line of the algorithm.

Our algorithm is essentially a repeated application of Theorem~\ref{thm:adjusted-flow}. We start with all prices equal to one and $f$ equal to a balanced flow.  In each iteration we choose appropriate values of $S$ and $x$ and then adjust prices and flow as given in (\ref{new prices}) to (\ref{new flow s}).

The price and flow adjustment decreases the surplus of type 2 buyers, increases the surplus of type 3 buyers and leaves the surplus of type 4 buyers unchanged. Thus there is the chance of the surplus of a type 2 buyer and a type 3 or 4 buyer to become equal. The effect on the norm of the surplus vector will be large if these surpluses were significantly different before the flow adjustment. These considerations led us to the following choice of $S$. We first determine the smallest  $\ell$ for which $\frac{r(b_\ell)}{r(b_{\ell+1})}>1+1/n$, and let $\ell=n$ when there is no such $\ell$. We then set $S=r(b_\ell)$ and obtain $B(S) = \sset{b_1,\ldots,b_\ell}$. With this choice the surpluses of the buyers in $B(S)$ are at least by a factor $1 + 1/n$ larger than the surpluses of the buyers outside $B(S)$ and two such surpluses becoming equal will constitute significant progress. If a new equality edge arises, the price adjustment by $x$ will not balance a big surplus and a small surplus. However, we will be able to use the new equality edge for balancing a big surplus and a small surplus. For technical reasons, we will never increase $x$ beyond $\xmax = 1 + 1/(Rn^3)$. 

When the total surplus becomes tiny (less than $\epsilon$), we exit from the loop. At this point, the price vector $p$ can be turned into a vector of equilibrium prices essentially by a rounding procedure. This will be discussed in Section~\ref{sec: rounding after termination}.

\begin{figure}[t]
\centerline{\framebox{\parbox{5.0in}{
\begin{tabbing}
555\=555\=555\=555\=\kill
\>\parbox{0.9\textwidth}{Set $R = 256$, $\epsilon={1}/(8n^{4n}U^{3n})$, $\Delta=n^5/\epsilon$, and
$\xmax = 1 + 1/(Rn^3)$;}\\[0.3em]
\>Set $p_i=1$ for all $i$ and compute a balanced flow $f$ in $N_p$; \\[0.3em]
\>{\bf Repeat}\\[0.3em]
\>\>Sort the buyers by their surpluses in decreasing order: $b_1,b_2,...,b_n$; \\[0.3em]
\>\>\parbox{0.8\textwidth}{Find the smallest  $\ell$ for which ${r(b_\ell)}/{r(b_{\ell+1})}>1+1/n$, and\\ 
\mbox{}\hfill let $\ell=n$ when there is no such $\ell$;}\\[0.3em]
\>\>Let $S=r(b_\ell)$ and $B(S) = \sset{b_1,\ldots,b_\ell}$; \\[0.3em]
\>\>Compute $x = \min(\xeq(S),x_{23}(S),x_{24}(S),\xmax)$;\\[0.3em]
\>\>compute $p'$ and $f'$ according to (\ref{new prices}) to (\ref{new flow s});\\[0.3em]
\>\>{\bf If} $x = \xeq$, modify $f'$ further according to the procedure in Figure~\ref{fig:aug};\\[0.3em]
\>\>$p = p'$ and $f = f'$;\\[0.3em]
\>\>update $f$ to a balanced flow in $N_p$; /* goods with zero surplus keep zero surplus */\\[0.3em]
\>{\bf Until} $|r(B)|<\epsilon$; \\[0.3em]
\>Round $p$ to equilibrium prices by the procedure in Figure~\ref{fig: final rounding};
\end{tabbing}
}}}\caption{The whole algorithm}\label{fig:algo}
\end{figure}

\begin{lemma}\label{pro:unchanged} Once the surplus of a good becomes zero, it stays zero. As long as a good has non-zero surplus, its price stays at one.
\end{lemma}
\begin{proof} Augmentation of a flow to a maximum flow does not increase any surplus. Balancing a flow leaves the surplus of all goods unchanged (Lemma~\ref{thm:compute-balanced}). Increasing prices and flows according to Theorem~\ref{thm:adjusted-flow} does not change the surplus of any good. Neither does the augmentation along a new equality edge according to the procedure in Figure~\ref{fig:aug}. \end{proof}

\begin{lemma}\label{r(bell) is large} Let $b_1,b_2,\ldots,b_n$ be the buyers sorted in order of decreasing surpluses. Let  $\ell$ be minimal such that $r(b_\ell)/r(b_{\ell + 1} > 1 + 1/n$. Let $\ell = n$, if there is no such $\ell$. Then $r(b_\ell) \ge \abs{r(B)}/(en)$, $r(b_i) \le e n r(b_\ell)$ for all $i$, and $r(b_i) \ge r(b_\ell)$ for every type 2 buyer $b_i$. 
\end{lemma}
\begin{proof} Clearly, $b_1$ is at least $|r(B)|/n$. Also, $\frac{r(b_j)}{r(b_{j+1})}\leq 1+1/n$ for $j<\ell$, and hence $r(b_\ell)\geq r(b_1)(1+1/n)^{-n}>|r(B)|/(e\cdot n)$. If $b_i$ is a type 2 buyer, $i \le \ell$. \end{proof}

\begin{lemma} With $S$ as defined above and $x = \min(\xeq(S),x_{23}(S),x_{24}(S),\xmax)$, $f'$ is a feasible flow in $N_{p'}$. \end{lemma}
\begin{proof} We only need to show that the surpluses of all buyers are nonnegative. If there is a type 3 or 4 buyer, this is obvious, because $x$ is at most the minimum of $x_{23}(S)$ and $x_{24}(S)$. If there are no type 3 or 4  buyers, $B(S)=B$, and each $r(b_i)$ is larger than $|r(B)|/(e\cdot n)$. Since goods in $\Gamma(B(S))$ have surplus zero and the total surplus is positive, there must be goods outside $\Gamma(B(S))$. Goods not in $\Gamma(B(S))$ have no buyers and hence positive surplus. Thus their price is one, and their surplus is also one. Hence $|r(B)| \ge 1$ and the price of the good owned by  a type 2 buyer is one. Thus the surplus of a type 2 buyer cannot reach zero, when $x = \xmax$. \end{proof}

In order to bound the number of iterations, we distinguish between $\xmax$-iterations ($x = \xmax$) and balancing iterations ($x < \xmax)$. In order to bound the number of $\xmax$-iterations, we first prove an upper bound on the maximum price (Lemma~\ref{thm:largest}) and then combine this upper bound with the observation that an $\xmax$-iteration increases at least one price by a factor $\xmax$ (Lemma~\ref{lem:number-of-bad}).
In order to bound the number of balancing iterations, we study the evolution of the norm of the surplus vector. We show that an $\xmax$-iteration increases it by at most a factor $1+O(1/n^3)$ and a balancing iteration divides it by at least a factor
of $1+\Omega(1/n^3)$. Since we have already a bound on the number of bad iterations, a bound on the number of good iterations follows (Lemma~\ref{lem:number-of-good}). 

\begin{lemma}\label{thm:largest} In the course of the algorithm, all prices stay bounded by $(nU)^{n-1}$.
\end{lemma}
\begin{proof}
It is enough to show that during the entire algorithm, for any non-empty and proper subset $\hat{C}$ of goods, there are goods $c_i\in\hat{C}, c_j\notin\hat{C}$ such that $p_i/p_j\leq nU$. So, when we sort all the prices in decreasing order, the ratio of two adjacent prices is at most $nU$. Since there is always a good with
price 1, the largest price is at most $(nU)^{n-1}$.

If $\hat{C}$ contains a good with surplus, then it contains a good with price one. The claim follows.

So assume all goods in $\hC$ have surplus zero, and let $\hat{B}=\Gamma(\hat{C})$ be the set of buyers adjacent to
goods in $\hat{C}$ in the equality graph.

If there exist $b_i \in \hB$ and $c_j \in C \setminus \hC$ with
$u_{ij}>0$, let $c_k\in \hC$ be one of the goods adjacent to $b_i$ in the
equality graph. Then $u_{ij}/p_j\leq \alpha_i = u_{ik}/p_k$. So,
$p_k/p_j\leq u_{ik}/u_{ij}\leq U$.

Otherwise, $u_{ij} = 0$ for every $b_i \in \hB$ and
$c_j \in C \setminus \hC$, i.e., the buyers in $\hB$ are only interested in goods in $\hC$. Then $\hB$ must be a proper subset of $B$ as otherwise no buyer would be interested in the goods in $C \setminus \hC$, a contradiction to Assumption (5). Since the buyers in $\hB$ must be interested in at least one good that is not owned by them (Assumption (6)), there must be a $k$ such that $c_k \in \hC$ and
$b_k \not\in \hB$.
Let $B' = \set{j }{b_j \in\hat{B}, c_j \not\in \hat{C}}$ and $B'' = \set{j}{b_j \not\in
\hat{B}, c_j \in \hat{C} }$. We have:
\begin{align*}
p_k &\le p(B'') = p(\hat{C}) - p(\set{j}{b_j \in \hat{B}, c_j \in \hat{C}}) \\
         & \le p(\hat{B}) - p(\set{j}{ b_j \in \hat{B}, c_j \in \hat{C}}) = p(B').
\end{align*}
The first inequality follows from $k \in B''$, and the second inequality holds since goods in $\hC$ have surplus zero and the flow into $\hC$ comes from $\hB$. We now have established
$p(B') \ge p_k > 0$. Thus, $B'$ is non-empty, and there is a $j
\in B'$ with $p_j \geq p(B')/n$. We conclude $p_k \leq n p_j$.
\end{proof}

We immediately obtain a bound on the number of $\xmax$-iterations.

\begin{lemma}\label{lem:number-of-bad} The number of $\xmax$-iterations is $O(n^5 \log (nU))$. 
\end{lemma}
\begin{proof} By Lemma~\ref{thm:largest}, prices are bounded by $(nU)^n$ and hence every price can be multiplied by $\xmax$ at most $O(\log_{\xmax}(nU)^{n})=O(n^4\log (nU))$ times. We conclude that the number of $\xmax$-iterations is $O(n^5\log(nU))$. \end{proof}

In order to bound the number of balancing iterations, we study the evolution of the norm of the surplus vector.

\begin{lemma}\label{thm:main} Let $f$ be a balanced flow in $N_p$ at the beginning of an iteration, let $f'$ be the flow in $N_{p'}$ constructed in the iteration, and let $r'(B)$ be the surplus vector with respect to $f'$. Then
\begin{compactenum}[\hspace{\parindent}--]
\item $\norm{r'(B)} \le (1 + O(1/n^3)) \cdot \norm{r(B)}$ in an $\xmax$-iteration ($x = \xmax$), and 
\item $\norm{r'(B)} \le \norm{r(B)}/(1 +\Omega(1/n^3))$ in a balancing iteration $(x < \xmax$).
\end{compactenum}
\end{lemma}
\begin{proof} We postpone the proof to the end of this section. \end{proof}

We can now prove a bound on the number of balancing  iterations.

\begin{lemma}\label{lem:number-of-good} The number of balancing iterations is $O(n^5 \log (nU))$.\end{lemma}
\begin{proof} For the initial flow and all prices equal to one, we have $\norm{r(B)} \le  \sqrt{n}$.
In the $\xmax$-iterations, the norm of the surplus vector is multiplied by a factor of at most
$(1+O(1/n^3))^{O(n^5\log (nU))}$. When the loop
terminates, $\norm{r(B)}<\epsilon$. Thus the number of balancing iterations is bounded by
\[ \log_{1+\Omega(1/n^3)}(\frac{1}{\epsilon}\sqrt{n}(1+O(1/n^3))^{O(n^5\log
(nU))}) =O(n^5\log (nU)). \]
\end{proof}

\begin{theorem} The algorithm in Figure~\ref{fig:algo} computes an equilibrium price vector with $O(n^9 \log (nU))$ arithmetic operations. The number of multiplications and divisions is only $O(n^8 \log (nU))$. \end{theorem}
\begin{proof} The number of iterations is $O(n^5\log (nU))$. In each iteration, we need to compute the quantities $\xeq(S)$, $x_{23}$, and $x_{24}(S)$. This requires $O(n^2)$ arithmetic operations. The flow update requires the same number of arithmetic operations. The computation of the new balanced flow requires $n$ max flow computations. Each max flow computation requires $O(n^3)$ arithmetic operations (only additions and abstractions). 

The extraction of the equilibrium prices after termination of the loop requires $O(n^4 \log (nU))$ arithmetic operations (Theorem~\ref{thm:termination}). \end{proof}

Before we prove Lemma~\ref{thm:main}, we give the details of how to adjust the flow when a new equality edge comes into existence, see Figure~\ref{fig:aug}. Let $(b_i,c_j)$ be the new equality edge. We first use the surplus of $b_i$ to decrease the surplus of $c_j$. Remaining surplus is then used to increase the surplus of other neighbors of $c_j$; this will further reduce the surplus of $b_i$. We stop as soon as the surplus of $b_i$ becomes equal to the surplus of a buyer in $B \setminus B(S)$ or, if $B(S) = B$, becomes zero.

\begin{figure}[t]
\centerline{\framebox{\parbox{\textwidth}{
\begin{tabbing}
555\=555\=555\=555\=\kill
\>Denote the new equality edge by $(b_i,c_j)$ ($b_i\in
B(S),c_j\notin\Gamma(B(S))$). \\[0.3em]
\>Let $w$ be the largest surplus of a buyer in
$B\setminus B(S)$. Let $w=0$ if $B(S)=B$.\\[0.3em]
\>($f''$ is the current flow during augmentation, and $r''$ is the surplus vector of $f''$.)\\[0.3em]
\>Augment along $(b_i,c_j)$ gradually until: \\[0.3em]
\>\> $r''(b_i)=w$ or $r''(c_j)=0$; \\[0.3em]
\>If $r''(b_i)=w$ then Exit; \\[0.3em]
\>For all buyers $b_k$ with flows to $c_j$ in $f'$ in any order \\[0.3em]
\>\>Augment along $(b_i,c_j,b_k)$ gradually until: \\[0.3em]
\>\>\>$r''(b_i)=\max (r''(b_k),w)$ or $f''(b_k,c_j)=0$; \\[0.3em]
\>\>Set $w=\max(r''(b_k),w)$; \\[0.3em]
\>\>If $r''(b_i)=w$ then Exit.\\[0.3em]
\>\>$f' = f''$
\end{tabbing} 
}}} \caption{Augmentation in case of a new equality edge}\label{fig:aug}
\end{figure}

\begin{lemma}\label{thm:case3a}
If there is a new equality edge $(b_i,c_j)$ with $b_i\in
B(S),c_j\notin\Gamma(B(S))$, the program in Figure~\ref{fig:aug} constructs $f''$ from $f'$
in which either $r''(b_i)=r'(b_i)-p_j$, or there is a $b_k\notin
B(S)$ with $r''(b_i)=r''(b_k)$, or $B(S) = B$ and $r''(b_i) = 0$. Moreover, the surplus of no good is increased.
\end{lemma}
\begin{proof}
During the procedure, the surplus of $b_i$ decreases and the surpluses of the buyers in $B \setminus B(S)$ may increase. However, we make sure that the surplus of $b_i$  does not
become less than the surplus of a buyer in $B\setminus B(S)$.  In the end, if $r''(b_i)=w$, then there is a
$b_k\in B\setminus B(S)$ such that $r''(b_i)=r''(b_k)$ or $B(S) = B$ and $r''(b_i) = 0$; otherwise,
$c_j$ has no surplus at the end, the entire flow to it comes from $b_i$, and there was no flow from $b_i$ to $c_j$ before the change of flow. Thus $r''(b_i)=r'(b_i)-p_j$.
\end{proof}

We come to the proof of Lemma~\ref{thm:main}. Let us summarize the effect of an iteration. \smallskip\begin{compactenum}
    \item[Property (1):] $x\leq \xmax$.\smallskip
    \item[Property (2):]: In $f'$, $r'(b)\geq r'(b')$ for all $b\in
    B(S), b'\notin B(S)$. Here, $r'(b)$ is the surplus of $b$
    w.r.t. $f'$, the flow corresponding to $x$ by
    Theorem~\ref{thm:adjusted-flow}.\smallskip
    \item[Property (3):] If $x<\xmax$, the following possibilities arise:
    \begin{enumerate}
        \item If $x = \xeq(S)$, there is a new equality edge $(b_i,c_j)$ with
        $b_i\in B(S),c_j\notin\Gamma(B(S))$. The procedure in Figure~\ref{fig:aug} yields a flow
        $f''$ in which either $r''(b_i)=r'(b_i)-p_j$, or $B(S) = B$ and $r''(b_i) = 0$, or there is a
        $b_k\notin B(S)$ with $r''(b_i)=r''(b_k)$ (same as (b)).
        \item If $x = x_{23}(S)$ or $x = x_{24}(S)$, there exist $b\in
    B(S)$ and $b'\notin B(S)$ such that $r'(b)=r'(b')$.
    \end{enumerate}
\end{compactenum}\smallskip

From Theorem~\ref{thm:adjusted-flow}, the surpluses in $f'$ will
increase for type 1 and 3 buyers, will decrease for type 2 buyers,
and will stay unchanged for type 4 buyers. The set $B(S)$ consists of the type 1 and type 2 buyers and hence before the change of flow, the surplus of any type 1 or type 2 buyer is larger than the surplus of any type 3 or type 4 buyer. The algorithm guarantees that the surplus
of a type 1 or 2 buyer cannot become smaller than the surplus of any
type 3 or 4 buyer. From Theorem~\ref{thm:adjusted-flow} and
Lemma~\ref{thm:case3a}, we infer that the surplus of no good increases and hence the total surplus does not increase, that 
the surpluses of type 2 and type 3 and 4 buyers move towards each other, and that $r'(b)=x\cdot
r(b)$ for any buyer $b$. Thus $\|r'(B)\|\leq
x\|r(B)\|=(1+O(1/n^3))\|r(B)\|$ and we have established the first claim of Lemma~\ref{thm:main}.\medskip

We next show that the second claim holds in cases (3a) and (3b). In (3a), there is a new equality edge $(b_i,c_j)$. Assume first that $r''(b_i) = r'(b_i) - p_j$ ($p_j \ge 1$). For all $b_k\notin B(S)$,
$r''(b_i)\ge r''(b_k)$, and $r''(b_k)=r'(b_k)+\delta_k$, where
$\delta_k\geq 0$ and $\sum_{b_k\notin B(S)}\delta_k\leq p_j$.
Because $|r(B)|\leq n$, $\|r(B)\|^2\leq n^2$. By
Lemma~\ref{thm:l2norm},
\begin{align*}
\|r''(B)\|^2 &\leq  \|r'(B)\|^2-p_j^2 \\
&\leq  x^2\|r(B)\|^2-1 \\
&\leq  x^2\|r(B)\|^2-\frac{1}{n^2}\|r(B)\|^2 \\
&= (1-\Theta(1/n^2))\|r(B)\|^2.
\end{align*}
So, we have $\|r''(B)\|=(1-\Omega(1/n^2))\|r(B)\|$.\smallskip

If $B(S) = B$ and after the procedure in Figure~\ref{fig:aug} we have $r''(b_i)=0$,
then $|r(B)|\geq 1$ since there are goods that have no buyers, and $r(b_i)\geq |r(B)|/(e\cdot n)$.
When $b_i$ is a type 1 buyer, $r'(b_i)>r(b_i)$. When $b_i$ is a type 2 buyer,
$p_i=1$, and $r'(b_i)=(1-x)p_i+x\cdot r(b_i)$ by Theorem~\ref{thm:adjusted-flow}.
Thus, $r'(b_i)\geq r(b_i)-\frac{1}{Rn^3}\geq \frac{|r(B)|}{e\cdot n}-\frac{1}{Rn^3}$.
Also $|r(B)|^2\geq \|r(B)\|^2$. Therefore, 
\begin{align*}
\|r''(B)\|^2 &\leq  \|r'(B)\|^2-r'(b_i)^2 \\
&\leq  \xmax ^2\|r(B)\|^2- \left(\frac{|r(B)|}{e\cdot n}-\frac{1}{Rn^3}\right)^2\\
&\leq  \xmax ^2\|r(B)\|^2- \left(\frac{|r(B)|}{e\cdot n}-\frac{|r(B)|}{Rn^3}\right)^2\\
&\leq  \norm{r(B)}^2 \left(\left(1 + \frac{1}{Rn^3}\right)^2 - \left(\frac{1}{e\cdot n}-\frac{1}{Rn^3}\right)^2\right)\\
&= (1-\Theta(1/n^2))\|r(B)\|^2,
\end{align*}
and hence $\|r''(B)\|=(1-\Omega(1/n^2))\|r(B)\|$.\smallskip

If after the procedure described in Figure~\ref{fig:aug}, there is
$b_k\notin B(S)$ such that $r''(b_i)=r''(b_k)$, then we are in a
similar situation as in (3b), possibly with an even smaller total
surplus. So, we can prove this case by the proof of (3b).\medskip

We turn to (3b). Let $u_1,u_2,...,u_k$ and $v_1,v_2,...,v_{k'}$ be the
list of original surpluses of type 2 and 3 buyers, respectively.
Define $u=\min u_i ,v=\max v_j $, so $u_i\geq u$ for all $i$,
and $v_j\leq v$ for all $j$, and $u>(1+1/n)v$. After the price and
flow adjustments in Theorem~\ref{thm:adjusted-flow}, the list of
surpluses will be $u_1-\delta_1, u_2-\delta_2,...,u_k-\delta_k$
and $v_1+\delta'_1, v_2+\delta'_2,...,v_k+\delta'_{k'}$ (here
$\delta_i,\delta'_j\geq 0$ for all $i,j$). Moreover, there exist $I$ and $J$
such that $u_I-\delta_I=v_J+\delta'_J$, where $u_I-\delta_I$ is
the smallest among $u_i-\delta_i$, and $v_J+\delta'_J$ is the
largest among $v_j+\delta'_j$ by property (2).
Since the total surplus is non-increasing,
and the surpluses of type 1 buyers (if any) increase, the total surplus of type 2 and type 3 buyers will be non-increasing.
Thus, $-\sum_i\delta_i+\sum_j\delta'_j\le 0$, and hence $\sum_i \delta_i\ge \sum_j \delta'_j$.
Clearly, $\delta_I\leq \sum_i\delta_i$, and $\delta'_J\leq\sum_j\delta'_j$. Therefore, 

\begin{align*}
\sum_i(u_i-\delta_i)^2+&\sum_j(v_j+\delta'_j)^2-(\sum_iu_i^2+\sum_jv_j^2)  \\
&=
-2\sum_iu_i\delta_i+2\sum_jv_j\delta'_j+\sum_i\delta_i^2+\sum_j{\delta'}_j^2
\\
&\leq
-u\sum_i\delta_i+v\sum_j\delta'_j-\sum_i\delta_i(u_i-\delta_i)+\sum_j\delta'_j(v_j+\delta'_j)
\\
&\leq
-(u-v)\sum_i\delta_i-(u_I-\delta_I)\sum_i\delta_i+(v_J+\delta'_J)\sum_j\delta'_j
\\
&\leq -(u-v)\sum_i\delta_i \\
&\leq -(u-v)\max(\delta_I,\delta'_J)\\
&\leq -\frac{1}{2}(u-v)^2 \\
&< -\frac{1}{2(n+1)^2}u^2\\
&\le - \frac{4}{Rn^2} e^2 u^2,
\end{align*}
where the last inequality uses $R = 256 \ge 32 e^2$ and $n+1\leq 2n$. 

Let $w_1,w_2,...w_{k''}$ be the list of surpluses of type 1
buyers; all of them are at most $e \cdot u$ (Lemma~\ref{r(bell) is large}) After price adjustment,
the surpluses will be $x\cdot w_1,x\cdot w_2,...x\cdot w_{k''}$
from Theorem~\ref{thm:adjusted-flow}. Therefore,
\begin{align*}
 \sum_i(xw_i)^2
&\leq (1+\frac{1}{Rn^3})^2\sum_i w_i^2 \\
&\leq \sum_i w_i^2+(\frac{2}{Rn^3}+\frac{1}{R^2n^6})\cdot ne^2u^2 \\
&=\sum_i w_i^2+(\frac{2}{Rn^2}+\frac{1}{R^2n^5})e^2u^2.
\end{align*}

Combining both bounds and using $\norm{r(B)}^2 \le n e^2 u^2$ since $r(b) \le e u$ (Lemma~\ref{r(bell) is large}) for every buyer $b$, we obtain
\begin{align*}
\|r'(B)\|^2 & <  \norm{r(B)}^2 + (- \frac{4}{Rn^2} + \frac{2}{Rn^2}+\frac{1}{R^2n^5})e^2u^2\\
&= \norm{r(B)}^2 + (-\frac{2}{Rn^2}+\frac{1}{R^2n^5})e^2u^2\\
&\le \|r(B)\|^2+(-\frac{2}{Rn^2}+\frac{1}{R^2n^5})\frac{1}{n}\|r(B)\|^2 \\
&= \|r(B)\|^2-\frac{2}{Rn^3}\|r(B)\|^2+\frac{1}{R^2n^6}\|r(B)\|^2 \\
&= \|r(B)\|^2(1-\frac{1}{Rn^3})^2.
\end{align*}
This completes the proof of Lemma~\ref{thm:main}.

\subsection{Extraction of the Equilibrium Prices}\label{sec: rounding after termination}

We show how to extract equilibrium prices from a price vector $p$ with total surplus less than $\epsilon$. 
The extraction process requires arithmetic on rational numbers on polynomial bitlength.

\begin{figure}[t]
\centerline{\framebox{\parbox{5.0in}{
\begin{tabbing}
555\=555\=555\=555\=\kill
\>If there is no surplus with respect to $p$, return $p$;\\[0.3em]
\>add the edge $(b_i,c_i)$ for each $i$ to the undirected equality graph $F_p$ to obtain $F'_p$;\\[0.3em]
\>{\bf Repeat}\\[0.3em]
\>\>\parbox{0.8\textwidth}{for some connected component of $F'_p$ with no surplus node, increase prices and flows by a common factor, until a new equality edges arises and thus two components are joined;}\\[0.3em]
\>{\bf Until} \begin{minipage}[t]{0.8\textwidth}{every connected component of $F'_p$ contains a surplus node and hence a good with price one;}\end{minipage}\\[0.3em]
\> \parbox{0.9\textwidth}{/* \emph{ For simplicity of description, we assume that $F'_p$ consists of a single connected component.}\Ran{*/}}\\[0.3em]
\>Let $\Psi_1$ to $\Psi_L$ be the connected components of $F_p$\\[0.3em]
\>Set up the following system of linear equations:\\[0.3em]
\>\> The equation $p'_i =1$, where $c_i$ is a good with surplus.\\[0.3em]
\>\>{\bf For each} $\Psi_k$\\[0.3em]
\>\>\> \parbox{0.85\textwidth}{$\abs{\Psi_k} - 1$ linearly independent equations of the form $u_{i j'} p'_j =u_{ij} p'_{j'}$, where $c_j$ and $c_{j'}$ are goods in $\Psi_k$;}\\[0.3em]
\>\>{\bf For each} $\Psi_k$, but one\\[0.3em]
\>\>\> The equation $\sum_{b_i\in B\cap\Psi_k}p'_i-\sum_{c_i\in C\cap\Psi_k}p'_i= 0$;\\[0.3em]
\>\parbox{0.9\textwidth}{Let $p'_i = q_i/D$, $1 \le i \le n$, be a solution to this system where $D$ is the determinant of the coefficient matrix and $q_i$ is integral;}\\[0.3em]
\>/* \emph{$p'_i$ is the rational with denominator at most $(nU)^n$ closest to $p_i$.} */\\[0.3em]
\> return $q = (q_1,\ldots,q_n)$;
\end{tabbing}}}}
\caption{\label{fig: final rounding} Converting the final price vector $p$ into a set of equilibrium prices.}
\end{figure}

We need the concept of the undirected equality graph $F_p$ with respect to a price vector $p$. It has vertex set $B \cup C$. A buyer $b$ and a good $c$ are connected by an undirected edge if and only if $(b,c) \in E_p$.

\begin{theorem}\label{thm:termination}
Let $p$ be a price vector with total surplus at most $\epsilon$. The procedure in Figure~\ref{fig: final rounding} converts $p$ into a vector of equilibrium prices. It runs in time $\tilde{O}(n^4\log U)$.
\end{theorem}
\begin{proof} In the procedure, we first add the edge $(b_i,c_i)$ for each agent $i$ to the undirected equality graph
$F_p$ to obtain $F'_p$ and then modify the prices such that each connected component of $F'_p$ contains a good with surplus. For a connected component of $F'_p$, the sum of the prices on both sides are the same. For all components $\Phi$ of $F'_p$ with no surplus node, increase prices and flows by a common
factor until a new equality edge emerges; this will unite a component with no surplus with a component with surplus and thus reduce the number of components with no surplus by one. A new equality edge must emerge because the buyers in a component own exactly the goods in the component and by property (6) they must also receive utility from some good that they do not own. Repeat this process until all components in $F'_p$ have a surplus node. The total surplus is still less than $\epsilon$.

We may assume w.l.o.g.~that $F'_p$ becomes connected by this process. Otherwise, the following argument can be applied independently to each component of $F'_p$. 

We will next discuss the set-up of a linear system of equations. We show that it has full rank and that the price vector $p$ satisfies it. We then show that the solution to the same system with slightly modified right-hand sides yields equilibrium prices. 

Denote the set of connected components in the undirected equality graph $F_p$ (not $F'_p$)
by $\Lambda = \sset{\Psi_k }$.  Consider any component $\Psi_k$ in $F_p$. For any buyer $b_i$ in the component and any two equality edges $ (b_i,c_j)$ and $(b_i,c_{j'})$, we have the equation
\begin{equation}\label{intra-component}  u_{i j'} p_j =u_{ij} p_{j'}. \end{equation}
A subset of $|\Psi_k\cap C|-1$ of these equations is linearly independent, since the price vector restricted to a component is determined up to multiplication by a common factor.
The total number of linear independent equations for all components in $F_p$ is $n-|\Lambda|$.

Since there is no flow between components, for each component $\Psi_k$ in $F_p$, the money difference between buyers and goods in $\Psi_k$ is equal to the surplus difference. Thus, we have the equation
\begin{equation}\label{inter-component}  \sum_{b_i\in B\cap\Psi_k}p_i-\sum_{c_i\in C\cap\Psi_k}p_i=\epsilon_k, \end{equation}
where $\epsilon_k$ (positive or negative) comes from the surpluses of goods and buyers in the component. Hence, $\sum|\epsilon_k|\leq 2\epsilon \Ran{<} 4\epsilon$. We work with the bound of $4 \epsilon$, because there will be additional error when we redo the proof of the theorem in Section~\ref{sec: poly time}.
If $b_i$ and $c_i$ belong to distinct connected components $\Psi_j$ and
$\Psi_k$, the coefficient of $p_i$ is $+1$ in the equation of $\Psi_j$, $-1$ in
the equation for $\Psi_k$, and $0$ in all other equations. If $b_i$ and $c_i$
belong to the same connected component, the coefficient of $p_i$ is zero in all
equations. In other words, in the coefficient matrix for the equations (\ref{inter-component}) we have one row for each component and one column for each $i$; the column for $i$ is all zero if $b_i$ and $c_i$ belong to the same component and contains a $+1$ and a $-1$ otherwise.
 Assume now that there is a proper subset $\Lambda'$ of the components such that the equations corresponding to them are linearly dependent. Then, if $b_i$ or $c_i$ belongs to some
component in the subset, both of them do. Thus for each $i$, the union of the components in $\Lambda'$ contains either $b_i$ and $c_i$ or contains neither $b_i$ nor $c_i$. We conclude that the union of the components is also a union of connected components of $F'_p$. However, $F'_p$ consists of a single component, and hence no proper subset of the equations is linearly dependent. Therefore, any
$|\Lambda| -1 $ of these equations are linearly independent. Let $L = \abs{\Lambda}$.

For the sequel, it will be convenient to eliminate the equations of type (\ref{intra-component}). In the equation (\ref{inter-component}) for component $\Psi_k$, there must be at least one $i$ such that $c_i \in C \cap \Psi_k$ and $b_i \not\in B \cap \Psi_k$ and at least one $i$ such that $b_i \in B\cap \Psi_k$ and $c_i \not\in C \cap \Psi_k$. This holds since prices are at least one and $\abs{\epsilon_k} \le 2\epsilon$. After suitable renumbering of goods, we may assume that $c_k \in C \cap \Psi_k$ and $b_k \not\in B \cap \Psi_k$. We refer to $c_k$ as the representative good for component $\Psi_k$.
For every good $c_j$ in component $\Psi_k$, the price of $c_j$ is linearly related to the price of $c_k$, i.e.,
\begin{equation} \label{intra-component-alt} p_j = t_j p_k,  \end{equation}
where $t_j$ is a fraction whose numerator and denominator is a product of at most $n$ utilities. Substituting equations (\ref{intra-component-alt}) into the equations (\ref{inter-component}), we
obtain the following equation for component $\Psi_k$:
\begin{equation}\label{inter-component-after-substitution}
   \left(\sum_{c_j \in C \cap \Psi_k} t_j\right) p_k -  \sum_{1 \le i \le \abs{\Lambda}, i \not= k}\left(\sum_{b_j \in B \cap \Psi_k, c_j \in C \cap \Psi_i} t_j\right) p_i = -\epsilon_k.
\end{equation}
Let $M = (m_{ki})$ be the coefficient matrix of this system of equations. Then $m_{kk} = \sum_{b_j \in B \cap \Psi_k, c_j \in C \cap \Psi_i} t_j > 0$,
$m_{ki} = -\sum_{b_j \in B \cap \Psi_i} t_j \le 0$, and all column sums of $M$ are equal to zero. 

There is a good $c_i$ with non-zero surplus. Its price $p_i$ is equal to one. We may assume w.l.o.g.~that $c_i$ belongs to $\Psi_L$. Then $p_i = 1$ is equivalent to $p_L= 1/t_i$. We will next argue that the equations (\ref{inter-component-after-substitution}) for components $\Psi_1$ to $\Psi_{L - 1}$ and equation $p_L = 1/t_i$ are linearly independent. Let $M'$ be the coefficient matrix for this system; $M'$ can be obtained from $M$ by setting $m_{Li}$ to zero for $1 \le i < L$.

Assume, for the sake of a contradiction, that the rank of $M'$ is less than $L$. Then there is a nonzero vector $a^T = (a_1,\ldots,a_L)$ such that $a^T M'= 0$. Let $a_{k_0}$ be an entry of maximum absolute value among $a_1$ to $a_{L-1}$. We may assume $a_{k_0} > 0$. Let $\Lambda' = \set{k}{k < L \text{ and } a_k = a_{k_0}}$. We may assume $\Lambda' = \sset{1,\ldots,L'}$. Consider any $k \le L'$. Then
\begin{align*}
0 &= \sum_{1 \le h < L} a_h m_{hk} + a_L \cdot 0 \\
   & = a_{k_0} \sum_{1 \le h \le L} m_{hk} - a_{k_0}m_{Lk} +  \sum_{L' < h < L} (a_h - a_{k_0}) m_{hk}\\
   &= - a_{k_0}m_{Lk} -  \sum_{L' < h < L} (a_{k_0} - a_h) m_{hk},
\end{align*}
and hence $m_{hk} = 0$ for $L' < h \le L$. We will next show that we also have $m_{kh} = 0$ for $k \le L'$ and $h > L'$. Intuitively, this holds since $m_{hk} = 0$ for $h > L'$ and $k \le L'$ means that there are no goods in components $\Psi_1$ to $\Psi_{L'}$ whose owners are in components $\Psi_{L'+1}$ to $\Psi_L$ and hence the buyers in components $\Psi_1$ to $\Psi_{L'}$ cannot own goods in components in $\Psi_{L'+1}$ to $\Psi_L$, since the surplus of each component is small and the price of each good is at least $1$. The formal argument follows. Summing the equations for $1 \le k \le L'$, we obtain
\[     \sum_{k \le L'} m_{kk} p_k + \sum_{k \le L'}\sum_{i;\; i \not= k} m_{ki} p_i = \sum_{k \le L'} \epsilon_k .\]
Rearranging the second sum yields
\begin{align*}
\sum_{k \le L'} \epsilon_k &= \sum_{k \le L'} m_{kk} p_k + \sum_{i} \sum_{k \le L';\; k \not= i} m_{ki} p_i\\
&= \sum_{k \le L'} m_{kk} p_k + \sum_{i \le L'} \sum_{k \le  L';\; k \not= i} m_{ki} p_i + \sum_{i > L'} \sum_{k \le L';\; k \not= i} m_{ki} p_i \\
&= \sum_{i \le L'} m_{ii} p_i + \sum_{i \le L'} \sum_{k \not= i} m_{ki} p_i + \sum_{i > L'} \sum_{k \le L'} m_{ki} p_i \\
&= \sum_{i \le L'} \left( \sum_k m_{ki} \right) p_i + \sum_{i > L'} \sum_{k \le L'} m_{ki} p_i\\
&= \sum_{i > L'} \sum_{k \le L'} m_{ki} p_i.
\end{align*}
If some $m_{ki}$ with $k \le L' $ and $i > L'$ is nonzero, the right-hand side is less than or equal to $- 1/U^n$, a contradiction.
We have now shown that $m_{kh} \not= 0$ only if either $k$ and $h$ are less than or equal to $L'$ or both are larger. This implies that $\Psi_1 \cup \ldots \cup \Psi_{L'}$ is a union of connected components of $F_p'$. However, $F_p'$ consists of a single component, and hence this is impossible. Thus $M'$ has full rank.

Altogether we have established that the system consisting of the equations (\ref{intra-component}),  (\ref{inter-component}), and equation $p_i = 1$ are linearly independent. In matrix form, we can write this system as $A p = X$. The matrix $A$ is invertible and has integral entries bounded by $U$.

Consider the system $A p' = X'$ for a price vector $p'$, where $X'$ is the unit vector with a one in the row corresponding to the equation $p_i = 1$.
The solution will be a vector of rational numbers with a common
denominator $D\leq (nU)^n$ by Cramer's rule. Since
$||X|-|X'||<4\epsilon$, any difference
$|p'_i-p_i|$ is at most $4\epsilon\cdot
(nU)^n=1/(2n^{3n}U^{2n})$ by Cramer's rule. The vector $p'$ can be computed in time $\tilde{O}(n^4\log U)$ time by Theorem 5.12 in~\cite{GG03}; solution of an $n \times n$ linear system with integral entries bounded by $U$. 

All prices $p'_i$ are of the form $q_i/D$, where $q_i,D$
are integers and $D\leq (nU)^n$. So, 
\[ |p_i -\frac{q_i}{D}|=|p'_i-p_i|\leq 1/(2n^{3n}U^{2n})=\epsilon'/{D},\]
where $\epsilon'= {D}/(2n^{3n}U^{2n})\leq {1}/(2n^{2n}U^n)$. 
Consider any $b_i \in B$ and $c_j,c_k \in C$ and assume $u_{ij}/p_j \le u_{ik}/p_k$. Then,
\begin{align*}
u_{ij} q_k & \leq  u_{ij}(p_k D + \epsilon') \\
& \leq   u_{ik} p_j D + u_{ij} \epsilon' \\
& \leq  u_{ik} (q_j + \epsilon') + u_{ij} \epsilon' \\
& \leq  u_{ik}q_j + (u_{ik} + u_{ij})\epsilon' \\
& <  u_{ik}q_j + 1,
\end{align*}
and hence, $u_{ij} q_k \le u_{ik}q_j$ since $u_{ij} q_k$ and $u_{ik}q_j$ are
integral. We conclude that the equality edges with respect to $p$ are also equality edges with respect to $q$ and hence the edges of $N_p$ are also present in $N_q$.

Let $Z = \sum_i q_i$. Then $Z$ is the capacity of the cuts $(s,B\cup C\cup t)$ and $(s\cup B\cup C, t)$ in $N_q$. We will show that both cuts are min-cuts in $N_q$ and hence the price vector $q$ is a market equilibrium. The capacity of the cuts above in $N_p$ is at least $(Z-n\epsilon')/D$. If the cuts above are not min-cuts, there must be a cut of capacity at most $Z - 1$ in $N_q$ as the capacities in $N_q$ are integral. This cut is also a cut in $N_p$ and has capacity at most
$(Z-1)/{D}+2n{\epsilon'}/{D}=Z/D- (1-2n\epsilon')/{D}$ in $N_p$. Thus, any maximum flow in $N_p$ has surplus
at least
\[    \sum_i p_i - \left( \frac{Z}{D} - \frac{1 - 2n \epsilon'}{D}\right) \ge - \frac{n \epsilon'}{D} +  \frac{1 - 2n \epsilon'}{D} =  \frac{1 - 3n \epsilon'}{D}  > \epsilon,\]
a contradiction.
\end{proof}

We remark that each $p'_i$ is the rational with denominator at most $(nU)^n$ closest to $p_i$. Observe first 
that the distance of two distinct rational numbers of denominator at most 
$(nU)^n$ is a at least $(nU)^{2n}$ and hence is larger than
$2|p'_i-p_i|$. Since $p'_i$ is a rational number with denominator
at most $(nU)^n$, $p'_i$ is the rational number with denominator at most $(nU)^n$ nearest to $p_i$. One can compute $p_i'$ from $p_i$ by continued fraction expansion. Continued fraction expansion requires the floor-operation in addition to the basic arithmetic operations. 

\subsection{General Case}\label{sec:general}

We now drop assumption (6) of Section~\ref{sec:model}, i.e., there may be a proper subset $P$ of agents such that $u_{ij} = 0$ for all $i \in P$ and $j \not\in P$. We follow~\cite{Jain07,Devanur-Garg-Vegh}. 
Consider the liking graph of agents in which there is a directed edge from $i$ to $j$ iff $u_{ij}>0$.
If the graph is strongly connected, assumption (7) is satisfied. Otherwise, we determine a topological order of strongly connected components $P_1,P_2,...P_k$, in which there are only edges from a lower numbered
to higher numbered components. \emph{We need the assumption that if a strongly connected component consists of a single agent $i$, then $u_{ii} > 0$, i.e., the agent likes its own good.} Assumptions 1 to 5 and this assumption are necessary and sufficient for the existence of an equilibrium in the linear exchange model~\cite{Gale57,Gale76}.

We use the algorithm in Section~\ref{sec:algo} to compute the equilibrium for the agents in every strongly connected component $P_i$ ($i=1,2,...,k$). Note that the existence of a self-loop for components of size one guarantees the existence of a solution for such components. For $i=2,...,k$, multiplying the prices in $P_i$ by
$(U+1)\cdot \max \set{p_j}{j\in P_{i-1}}$ ensures that there are no equality edges from $P_i$ to $P_j$ for $i<j$. Since the agents in $P_j$ do not like any goods in $P_i$ for $i<j$, this will not affect the equilibrium of any component, and hence we obtain a global equilibrium. 

\subsection{A Remark on Arithmetic}

Our algorithm uses only rational arithmetic. Utilities are assumed to be integral and all prices are initially equal to one. In each iteration we multiply some of the prices by a factor $x$; $x$ is the minimum of $\xeq(S)$, $x_{23}(S)$, $x_{24}(S)$, and $\xmax$. Each of these quantities is rational. However, the bitlength of the prices may potentially double in each iteration, since $\xeq(S)$, $x_{23}(S)$, and $x_{24}(S)$ are quotients involving prices and surpluses.

\section{Polynomial Time}\label{sec: poly time}

In order to achieve polynomial time, we have to reduce the cost of arithmetic without loosing the polynomial bound on the number of iterations. We use the following approach. We restrict prices and the update factor $x$ to powers of $(1 + 1/L)$, where $L$ has polynomial bitlength. In the iterative part of the algorithm, we also approximate utilities by powers of $1 + 1/L$. We compute surpluses only approximately. The key technical result is that the approximate computation of surpluses and the factor $x$ increases the norm of the surplus vector by only a  factor $1 + O(1/n^4)$ in each iteration. This is an order of magnitude less than the change stated in Lemma~\ref{thm:main}. We conclude that the number of iterations is still $O(n^5 \log(nU))$. In the extraction of the equilibrium, we have to cope with the additional error introduced by rounding the utilities to powers of $(1 + 1/L)$. This is no problem, since we left some leeway in the proof of Theorem~\ref{thm:termination}.

\subsection{Iterative Improvement, Revisited}\label{sec: iterative improvement poly}

We restrict prices to the form $(1 + 1/L)^k$, $0 \le k \le K$, where $L = 16n^5 (nU)^n/\epsilon = 128 n^{5n + 5} U^{4n}$ and $K$ is chosen such that $(nU)^n \le (1 + 1/L)^K$. This choice of $K$ guarantees that the full range $[1..(nU)^n]$ of potential prices is covered. Then $K = O(n L \log(nU))$. We represent a price $p_i = (1 + 1/K)^{e_i}$ by its exponent $e_i \in \N$. Here, $\N$ denotes the natural numbers including zero. The bitlength of $e_i$ is $\log K = O(n \log (nU))$.

We approximate utilities by powers of $(1 + 1/L)$. For a utility $u_{ij} \in [1..U]$, let $e_{ij} \in \N$ be such that 
$\abs{e_{ij} - \log_{1 + 1/L} u_{ij}} < 1$ and let $\tu_{ij} = (1 + 1/L)^{e_{ij}}$. Then $u_{ij}/(1 + 1/L) \le \tu_{ij} \le u_{ij}(1 + 1/L)$. The exponent $e_{ij}$ can be computed from $u_{ij}$ in time polynomial in $\log U$ and $\log L$, see Section~\ref{sec: details of arithmetic}. It is tempting to define $e_{ij}$ as $\floor{\log_{1 + 1/L} u_{ij}}$, but for this definition, we do not know how to compute $e_{ij}$ in polynomial time.

We say that real number $b$ is an \emph{additive $1/L$-approximation} of real number $a$ if $\abs{a - b} \le 1/L$.
It is a \emph{multiplicative $(1 + 1/L)$-approximation} if $a/b \in [1/(1 + 1/L),1 + 1/L]$.

We compute the surplus vector with respect to $p$ only approximately. To this end, we replace each price $p_i$ by an approximation $\hp_i$ with denominator $L$ (note that the denominator of $p_i$ might be as large as $L^K$) and compute a balanced flow in a modified equality network $N(p,\hp)$. 

The approximation $\hp_i$ is a rational number with denominator $L$ and an additive $1/L$- and a multiplicative $(1 + 1/L)$-approximation of $p_i$, i.e.,
\begin{equation}\label{rounded prices} 
 \hp_i = \frac{q_i}{L}, \quad q_i \in \N,\quad \abs{p_i - \hp_i} \le \frac{1}{L}, \quad \frac{\hp_i}{p_i} \in [1/(1 + 1/L), (1 + 1/L)]. \end{equation}
We will see in Section~\ref{sec: details of arithmetic} that $q_i$ can be computed in time polynomial in $\log L$. Again, it is tempting to define $q_i$ as $\floor{p_i L}$, but for this definition, we do not know how to compute $q_i$ in polynomial time. 

\newcommand{\hq}{\hat{q}}

For two price vectors $p$ and $\hp$, the equality network $N(p,\hp)$ has its edge set determined by $p$ and its edge capacities determined by $\hp$, i.e., it has 
\begin{compactenum}[\hspace{\parindent}--]
    \item an edge $(s,b_i)$ with capacity $\hp_i$ for each $b_i \in B$,
    \item an edge $(c_i,t)$ with capacity $\hp_i$ for each $c_i \in C$, and 
    \item an edge $(b_i,c_j)$ with infinite capacity whenever $\tu_{ij}/p_j=\max_\ell \tu_{i\ell}/p_\ell$. We use $E_p$ to denote this set of edges.
\end{compactenum}\medskip

Let $\hf$ be a balanced flow in $N(p,\hp)$. For a buyer $b_i$, let $\hr(b_i) = \hp_i - \hf_{si}$ be its surplus, for a good $c_i$, let $\hr(c_i) = \hp_i - \hf_{it}$ be its surplus. As in Section~\ref{sec:algo}, we order the buyers by decreasing surplus and let $\ell$ be minimal such that $\hr(b_{\ell})/\hr(b_{\ell + 1}) > 1 + 1/n$. If there is no such $\ell$, let $\ell = n$. Let $S = \hr(b_\ell)$ and $B(S) = \sset{b_1,\ldots,b_\ell}$. Then $\hr(b_\ell) \ge \hr(b_1)/e \ge \hr(B)/(en)$ as in Section~\ref{sec:algo}. 

Since prices are now restricted to powers of $1 + 1/L$, the update factor $x$ has to be a power of $1 + 1/L$, and hence we need to modify its definition and computation. We compute $x$ in two steps. We first compute a factor $\hx$ from $\hp$ as in Section~\ref{sec:algo} and then obtain $x$ from $\hx$ by rounding to a near power of $1 + 1/L$. We use $x$ to update the price vector $p$. The prices of all goods in $\Gamma(B(S))$ are multiplied by $x$. The resulting algorithm is shown in Figure~\ref{fig: algo poly}.

\begin{figure}[t]
\centerline{\framebox{\parbox{5.0in}{
\begin{tabbing}
555\=555\=555\=555\=\kill
\>\parbox{0.9\textwidth}{Set $R = 32 e^2$, $\epsilon={1}/(4n^{4n}U^{3n})$, $\Delta=n^5/\epsilon$, and
$\xmax = 1 + 1/(Rn^3)$;}\\[0.3em]
\>Set $\hp_i = p_i=1$ for all $i$ and let $\hf$ be a balanced flow in $N(p,\hp)$;\\[0.3em]
\>{\bf Repeat}\\
\>\>Sort the buyers by their surpluses $\hr(b)$ in decreasing order: $b_1,b_2,...,b_n$; \\[0.3em]
\>\>\parbox{0.8\textwidth}{Find the smallest  $\ell$ for which $\hr(b_\ell)/\hr(b_{\ell+1})>1+1/n$, and\\ 
\mbox{}\hfill let $\ell=n$ when there is no such $\ell$;}\\[0.3em]
\>\>Let $S=\hr(b_\ell)$ and $B(S) = \sset{b_1,\ldots,b_\ell}$; \\[0.3em]
\>\>Compute $\hx = \min(\xeq(S),\hx_{23}(S),\hx_{24}(S),\xmax)$;\\[0.3em]
\>\>Compute $x$ as an $1 + 1/L$ approximation of $\hx$;\\[0.3em]
\>\>Replace $p_i$ by $x p_i$ for $c_i \in \Gamma(B(S))$;\\[0.3em]
\>\>Let $\hp$ be computed from $p$ according to (\ref{rounded prices});\\[0.3em]
\>\>Update $\hf$ to a balanced maximum flow in $N(p,\hp)$; \\[0.3em]
\> {\bf Exit from the loop if} $|\hr(B)|<\epsilon$; \\[0.3em]
\>Round $p$ to equilibrium prices by the procedure in Figure~\ref{fig: final rounding};
\end{tabbing}
}}}\caption{The polynomial time algorithm}\label{fig: algo poly}
\end{figure}

We define $\hx$ as the minimum of $\xeq(S)$, $\hx_{23}(S)$, $\hx_{24}(S)$,  and $\xmax$. The definition of $\xeq(S)$ is in terms of the rounded utilities:
\[\xeq(S)=\min \set{\frac{\tu_{ij}}{p_j}\cdot\frac{p_k}{\tu_{ik}}}{b_i\in B(S),
(b_i,c_j)\in E_p, c_k\notin\Gamma(B(S))}.\]
By definition, $\xeq(S)$ is a power of $1 + 1/L$. We redefine $\xmax$ as a power of $(1 + 1/L)$ such that 
$1 + 1/(Rn^3) \ge \xmax \ge (1 + 1/R(n^3))/(1 + 1/L)^2$. Such an $\xmax$ can be computed in polynomial time, see Section~\ref{sec: details of arithmetic}. The quantities $\hx_{23}(S)$ and $\hx_{24}(S)$ are defined in terms of the 
price vector $\hp$ and the surplus vector $\hr$. 
\begin{align*}
\hx_{23}(S) &= \min \set{\frac{\hp_i + \hp_j - \hr(b_j)}{\hp_i + \hp_j - \hr(b_i)}}{\text{$b_i$ is type 2 and $b_j$ is type 3 buyer}}\\
\hx_{24}(S) &= \min \set{\frac{\hp_i - \hr(b_j)}{\hp_i - \hr(b_{\Ran{i}})}}{\text{$b_i$ is type 2 and $b_j$ is type 4 buyer}}
\end{align*}
We then define $x_{23}(S)$ and $x_{24}(S)$ as powers of $1 + 1/L$ such that 
\[   \hx_{23}(S)/x_{23}(S) , \hx_{24}(S)/x_{24}(S)  \in [1/(1 + 1/L), 1 + 1/L].\]
The quantities $\xeq(S)$, $\hx_{23}(S)$, $\hx_{24}(S)$, $x_{23}(S)$, and $x_{24}(S)$ can be computed in polynomial time, see Section~\ref{sec: details of arithmetic}. Let $\hx = \min(\xeq(S),\hx_{23}(S),\hx_{24}(S),\xmax)$ and let $x = \min(\xeq,x_{23}(S),x_{24}(S),\xmax)$. Clearly, $x$ is a multiplicative $(1 + 1/L)$-approximation of $\hx$. 

We use $x$ to update the prices: $p'_i = x p_i$ for each good $c_i \in \Gamma(B(S))$ and 
$p'_i = p_i$ for any $c_i \not\in \Gamma(B(S))$. The new vector $\barp$ of rounded prices is defined by (\ref{rounded prices}) from $p'$. We also introduce an intermediate price vector $\tp$ which plays no role in the algorithm, but is crucial for the analysis: $\tp_i = \hx \hp_i$ for $c_i \in \Gamma(B(S))$ and $\tp_i = \hp_i$ for $c_i \not\in \Gamma(B(S))$. Observe that $\barp$ is obtained from $p$ by first updating to $p'$ by $x$ and then rounding according to (\ref{rounded prices}) and that $\tp$ is obtained from $p$ by first rounding to $\hp$ and then updating by $\hx$. Clearly, $\barp_i = \tp_i$ for $c_i \not\in \Gamma(B(S))$. For $c_i \in \Gamma(B(S))$, $\tp_i$ is a multiplicative $(1 + 1/\Delta)^3$ approximation of $\barp_i$, since
\begin{align*}
\frac{\tp_i}{\barp_i} &\in  \frac{ \hx \hp_i}{p'_i \cdot [1/(1 + 1/L),1 + 1/L]} \\
&\subseteq \frac{ x p_i \cdot [1/(1 + 1/L)^2,(1 + 1/L)^2]} { x p_i \cdot [1/(1 + 1/L),(1 + 1/L) ]}  \subseteq [1/(1 + 1/L)^3,(1 + 1/L)^3 ] .\end{align*}
Figure~\ref{fig: poly prices} illustrates the definitions.   

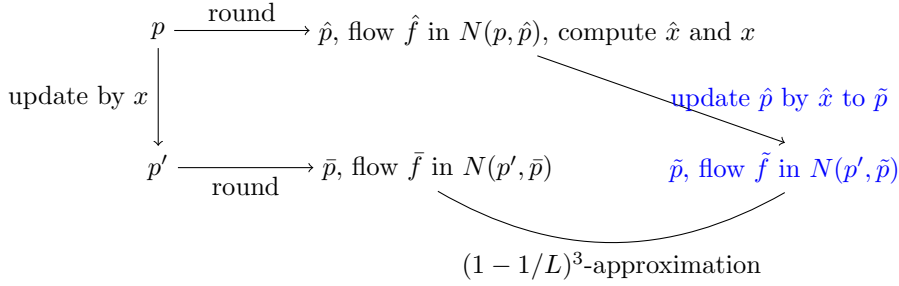
\begin{figure}[t]
\begin{center}
\begin{tikzpicture}[node distance=1.3cm]
\node (p) {$p$};
\node (pprime) [below=of p] {$p'$};
\node (phut) [right=1.8cm of p] {$\hat p$, flow $\hat f$ in $N(p,\hat p)$, compute $\hx$ and $x$};
\node (pbar) [right=1.8cm of pprime] {$\bar p$, flow $\bar f$ in $N(p',\bar p)$};
\node (pschlange) [right=of pbar] {\textcolor{blue}{$\tilde p$, flow $\tilde f$ in $N(p',\tp)$}};

\draw[->] (p) to node [left, align=center] {update by $x$}(pprime);
\draw[->] (p) to node [above] {round} (phut);
\draw[->] (pprime) to node [below] {round} (pbar);
\draw[-] (pbar.south) to [bend right] node [below] {$(1-1/L)^3$-approximation} (pschlange.south);
\draw[->] (phut.south) to node [right=0cm, align=left] {\textcolor{blue}{update $\hp$ by $\hx$ to $\tp$}} (pschlange.north);
\end{tikzpicture}
\end{center}
\caption{\label{fig: poly prices} We maintain $p$, $\hp$, $p'$ and $\barp$. We compute $\hx$ from $\hp$ and the balanced flow $\hf$ in $N(p,\hp)$; $x$ is obtained from $\hx$ by rounding 
to power for $1 + 1/L$. We update $p$ by $x$ to obtain $p'$ which we then round to $\barp$. The flow $\barf$ is a balanced flow in $N(p',\barp)$. We obtain it from $\hf$ by first multiplying the flow \Ran{on} the edges incident to goods in $\Gamma(B(S))$ by factors near $x$, then augmenting the flow as in Figure~\ref{fig:aug} and finally balancing the flow. The price vector $\tp$ and the flow $\tf$ are used in the analysis; $\tp$ is obtained from $\hp$ by updating with $\hx$ and $\tf$ is then obtained as in Section~\ref{sec:algo}.}
\end{figure}

Lemma~\ref{thm:main} applies to the transition from a balanced flow $\hf$ in $N(p,\hp)$ to a balanced flow $\tf$ in $N(p',\tp)$. In particular,
\begin{compactenum}[\hspace{\parindent}--]
\item $\norm{\tr(B)} \le (1 + O(1/n^3)) \cdot \norm{\hr(B)}$ in an $\xmax$-iteration ($x = \xmax$), and 
\item $\norm{\tr(B)} \le \norm{\hr(B)}/(1 +\Omega(1/n^3))$ in a balancing iteration $(x < \xmax$).
\end{compactenum}
Here $\hr(B)$ and $\tr(B)$ are the surplus vectors with respect to flow $\hf$ in $N(p,\hp)$ and flow $\tf$ in $N(p',\tp)$. 

Let $\barf$ be a balanced maximum flow in $N(p',\barp)$ with the property that goods with surplus zero with respect to $\hf$ also have surplus zero with respect to $\barf$.\footnote{One may compute $\barf$ as follows. First, determine for each good $c_i$ the factor $x_i = \barp_i/\hp_i$. Then multiply the flows on all edges incident to $c_i \in \Gamma(B(S))$ by $x_i$. Adjust the flow on the edges $(s,b_i)$ such that flow conservation holds. If $x = \xeq(S)$, 
run the procedure of Figure~\ref{fig:aug}. Next, augment the flow along augmenting paths from $s$ to $t$ until the flow is maximum. Finally, balance the flow.} Let $f^* = \tf/(1 + 1/L)^3$ be the flow $\tf$ scaled by a factor $1/(1 + 1/L)^3$. Since $\barp_i \ge \tp_i/(1 + 1/L)^3$, $f^*$ is a feasible flow in $N(p',\barp)$. Thus
\[   \norm{ \barr(B) } \le \norm { r^*(B) },\]
where $r^*(B)$ is the surplus vector with respect to $f^*$ and $\barr(B)$ is the surplus vector with respect to $\barf$. The inequality holds since $\barf$ is a balanced maximum flow in $N(p',\barp)$ and $f^*$ is a feasible flow in the same network. We finish the analysis of the update step by showing that $\norm { r^*(B) }$ is bounded by $(1 + O(1/n^4)) \cdot \norm{\tr(B)}$.


\begin{lemma} For any buyer $b_i$: $\abs{r^*(b_i) - \tr(b_i)} \le 8(nU)^n/L$. \end{lemma}
\begin{proof} By the definition of surpluses, 
\[ \abs{r^*(b_i) - \tr(b_i)} = \abs{\barp_i - f^*_{si} - (\tp_i - \tf_{si})} \le \abs{\barp_i - \tp_i} + \abs{f^*_{si} - \tf_{si}}.\]
In both terms on the right hand side, the quantities involved are bounded by the maximum price and are multiplicative $(1 + 1/L)^3$-approximations of each other. Thus the difference is bounded by $2 (nU)^n\cdot 4/L$, since
$(1 + 1/L)^3 - 1 \le 4/L$. \end{proof}

\begin{lemma}\label{lem: star and tilde} $\norm { r^*(B) } = (1 + O(1/n^4)) \cdot \norm{\tr(B)}$ if $\abs{\tr(B)} \ge \epsilon/2$ and $\abs{r^*(B)} \le \epsilon$ if $\abs{\tr(B)} \le \epsilon/2$. \end{lemma}
\begin{proof} Assume first, that $\abs{\tr(B)} \ge \epsilon/2$. Then (using $\norm{\tr(B)}^2 \ge \abs{\tr(B)}^2/n$)
\begin{align*}
\norm{r^*(B)} &\le \sum_{b \in B} (\tr(b) + 8 (nU)^n/L)^2\\
&= \norm{\tr(B)}^2+ 2 \cdot 8 (nU)^n\abs{\tr(B)}/L + 64 n (nU)^{2n}/L^2 \\
&\le \norm{\tr(B)}^2(1+\frac{1}{n^4})^2.
\end{align*}
The inequality holds since
\[ 8 (nU)^n\abs{\tr(B)}/L   \le     \abs{\tr(B)}^2/n^5 \le \norm{\tr(B)}^2/n^4 \]
and
\[ 8 n^{1/2} (nU)^{n}/L\le \abs{\tr(B)}/n^{9/2}  \le \norm{\tr(B)}/n^4.  \]
Assume next that $\abs{\tr(B)} \le \epsilon/2$. Then
\[   \abs{r^*(B)} \le \abs{\tr(B)} + 8 n (nU)^n/L \le \epsilon.\]
\end{proof}

We can now reprove Lemma~\ref{thm:main} and the polynomial bound on the number of iterations. 

\begin{lemma} Let $\hf$ be a balanced maximum flow in $N(p,\hp)$ and let
$\barf$ be a balanced maximum flow in $N(p',\barp)$. Then either $\abs{\barr(B)} \le \epsilon$ or
\begin{compactenum}[\hspace{\parindent}--]
\item $\norm{\barr(B)} = (1 + O(1/n^3)) \norm{\hr(B)}$ in an $\xmax$-iteration, and
\item $\norm{\barr(B)} = \norm{\hr(B)}/(1 + \Omega(1/n^3))$ in a balancing iteration.
\end{compactenum}
\end{lemma}
\begin{proof} Assume $\abs{\barr(B)} > \epsilon$. Then $\abs{r^*(B)} > \epsilon$ since $\barf$ is a maximum flow and $f^*$ is a feasible flow. Thus $\tr(B) > \epsilon/2$ and $\norm { r^*(B) } = (1 + O(1/n^4)) \cdot \norm{\tr(B)}$
by Lemma~\ref{lem: star and tilde}. The claim follows.\end{proof}

\begin{theorem}\label{thm: poly bound} The number of iterations is $O(n^5 \log(nU)$. \end{theorem}
\begin{proof} This is proved as in Section~\ref{sec:algo}.\end{proof}

\subsection{Extracting the Market Equilibrium, Revisited}\label{sec: rounding after termination poly}

We proceed essentially as in Section~\ref{sec: rounding after termination}. We go through the proof of Theorem~\ref{thm:termination} and indicate the required changes. 

We first make sure that every component of $F'_p$ contains a good with surplus and hence price one. This is
as in the proof of Theorem~\ref{thm:termination}. We work with the rounded utilities $\tu$ and prices that are powers of $1 + 1/L$. 

We next discuss the set-up of the linear system. Consider any component $\Psi_k$ of the undirected equality graph (with respect to utilities $\tu$ and price vector $p$). As in Section~\ref{sec: rounding after termination}, select any $\abs{\Psi_k \cap C} - 1$ independent equations. For each equation $\tu_{i j'} p_j =\tu_{ij} p_{j'}$ in this set add the equation
\[   u_{ij'} p_j - u_{ij} p_{j'} = (u_{ij'} - \tu_{ij'}) p_j + (\tu_{ij} - u_{ij})p_{j'} \]
to the system. Note that the absolute value of the right hand side is bounded by $2 (nU)^n U/L \le \epsilon/n$. 

Since there is no flow between components, for each component $\Psi_k$ in $F_p$, the money difference between buyers and goods in $\Psi_k$ is equal to the surplus difference. Thus, we have the equation
\[   \sum_{b_i\in B\cap\Psi_k}\hp_i-\sum_{c_i\in C\cap\Psi_k}\hp_i=\epsilon_k \]
where $\epsilon_k$ (positive or negative) comes from the surpluses of goods and buyers in the component. Hence, $\sum|\epsilon_k|\leq 2\epsilon$. We replace $\hp$ by $p$ and add the equation
\[ \sum_{b_i\in B\cap\Psi_k} p_i-\sum_{c_i\in C\cap\Psi_k} p_i=\epsilon_k + \delta_k \]
to the system. Here $\abs{\delta_k} \le n (nU)^n/L  \le \epsilon/n$. 

The last equation is $p_i = 1$, where $c_i$ is a good that had price one at termination of the loop. 

The linear system $A p = X$ has full rank. As above, let $X'$ be the unit vector with a zero in the row corresponding to the equation $p_i = 1$, and let $p'$ be a solution to the system $A p' = X'$. As before, $p'$ is a vector of rationals with common denominator $\det A$. Let $p'_i = q_i/D$ with $q_i \in \N$. Observe that $\abs{|X| - |X'|} < 4 \epsilon$. 
In the proof of Theorem~\ref{thm:termination} we used the same inequality. 

We next show that equality edges with respect to utilities $\tu$ and price vector $p$ are equality edges with respect to the true utilities $u$ and price vector $q$. Observe
\[    |p_i -\frac{q_i}{D}|=|p'_i-p_i|\leq 1/(2n^{3n}U^{2n})=\epsilon'/{D},\]
where $\epsilon'= {D}/(2n^{3n}U^{2n})\leq {1}/(2n^{2n}U^n)$. 
Consider any $b_i \in B$ and $c_j,c_k \in C$ and
  assume $\tu_{ij}/p_j \le \tu_{ik}/p_k$. Then,
\begin{align*}
u_{ij} q_k & \leq  \tu_{ij} (1 + 1/L) (p_k D + \epsilon') \\
& \leq   \tu_{ij} p_k D  + \tu_{ij}((p_k D + \epsilon')/L + \epsilon')\\
& \le \tu_{ik} p_j D  + \tu_{ij}((p_k D + \epsilon')/L + \epsilon')\\
&\le u_{ik}(1 + 1/L)(q_j + \epsilon') + \tu_{ij}((p_k D + \epsilon')/L + \epsilon')\\
& \leq  u_{ik}q_j + u_{ik} ((q_j + \epsilon')/L + \epsilon') + \tu_{ij}((p_k D + \epsilon')/L + \epsilon')\\
&<  u_{ik}q_j +  3 U (nU)^{2n}/L \\
& \le  u_{ik}q_j + 1,
\end{align*}
and hence, $u_{ij} q_k \le u_{ik}q_j$ since $u_{ij} q_k$ and $u_{ik}q_j$ are
integral. Thus every equality edge with respect to utilities $\tu$ and price vector $p$ is also an equality edge with respect to the true utilities $u$ and price vector $q$. Let $N_q$ be the flow network with respect to price vector $q$. The proof is now completed as in Section~\ref{sec: rounding after termination}. 

\begin{theorem}\label{thm:termination poly} Assume $\hr(B) \le \epsilon$. Then a vector of equilibrium prices can be extracted from the equality graph $E_p$ with respect to the rounded utilities $\tu$ in time $\tilde{O}(n^4 \log(nU))$. \end{theorem}

\subsection{Details of Arithmetic}\label{sec: details of arithmetic}

We show that arithmetic in rational numbers of bitlength $O(n \log(nU))$ suffices. 

\begin{lemma} Let $a \ge 1$ be a real number and let $b$ be a rational with $\abs{b - a} \le 1/(4L)$. Let $L \ge 3$ and let $q$ be the integer closest to $bL$. Then 
\[  \abs{a - \frac{q}{L}} \le \frac{3}{4L} < \frac{1}{L} \quad\text{and}\quad \frac{a}{q/L} \in [\frac{1}{1 + 1/L},1 + 1/L].\]
\end{lemma}
\begin{proof} Clearly $\abs{q - bL} \le 1/2$ and hence $\abs{q/L - a} \le \abs{q/L - b} + \abs{b - a} \le 3/(4L)$. This proves the first claim. For the second claim, we divide by $a$ and obtain $\abs{(q/L)/a - 1} \le 3/(4aL) \le 3/(4L)$. For $L \ge 3$, $1/(1 + 1/L) \ge 1 - 3/(4L)$ and $1 + 3/(4L) \le 1 + 1/L$. \end{proof}

\begin{lemma} Let $k \in \N$ be such that $(1 + 1/L)^k \le (nU)^n$. Then $\log k = O(nL \log(nU)\Ran{)}$. A rational number $b$ with $\abs{b - (1 + 1/L)^k)} \le 1/(4L)$ can be computed with $O(n \log(nU))$ arithmetic operations (additions, subtractions, multiplications) on rational numbers of bitlength $O(n \log(nU))$. \end{lemma}
\begin{proof} We compute $(1 + 1/L)^k$ by repeated squaring using fixed point arithmetic with $Z$ bits after the binary point. Then the rounding error in any arithmetic operation is at most $2^{-Z}$. Any rounding error is blown up by a factor at most $(nU)^n$ by subsequent operations, since all intermediate values are bounded by $(nU)^n$. 
Repeated squaring requires $\log k$ operations. Thus the total error is $O(nL \log (nU)\Ran{)} \cdot (nU)^n \cdot 2^{-Z}$. 
This is less than $1/(4L)$ for $Z = \Omega(n \log(nU) + \log L)$. \end{proof}

The mapping from $p$ to $\hp$ and from $p'$ to $\barp$ is covered by the preceding Lemma. It requires $n^2 \log(nU)$ arithmetic operations on numbers of bitlength $O(n \log(nU))$. 

The computation of $\xeq$ requires $O(n^2)$ additions and subtractions of numbers of bitlength $O(n \log(nU))$. 

The components $\hf_{ij}$ of the balanced flow $\hf$ are rational numbers bounded by $(nU)^n$ with denominator $Ln!$. Thus, $\hx_{23}$ and $\hx_{24}$  \Rantwo{$n^2$}{} can be computed with $n^2$ arithmetic operations on 
numbers of bitlength $O(n\log (nU))$. 

\begin{lemma} Let $\hx$ be a rational number. Then a $k \in \N$ such that $ x = (1 + 1/L)^k$ is a multiplicative $(1 + 1/L)$-approximation of $\hx$ can be computed with $O(\log L)$ arithmetic operations (additions, subtractions, multiplications) on numbers of bitlength $O(n \log(nU))$ \end{lemma}
\begin{proof} We can find $k$ by binary search. In each step of the binary search, we need to compare $\hx$ with $(1 + 1/L)^\ell$ for some $\ell$. Assume that $b$ is a multiplicative $(1 + 1/L)$-approximation of $(1 + 1/L)^\ell$ and let $k' = \log_{1 + 1/L} \hx$. We compare $\hx$ with $b(1 + 1/L)$ and $b/(1 + 1/L)$. If $\hx > b(1 + 1/L)$, $k' > \ell$. 
If $\hx < b/(1 + 1/L)$, $k' < \ell$. If $b/(1 + 1/L) \le \hx \le b(1 + 1/L)$, we may return $\ell$. Also, once we have found $\ell$ such that $(1 + 1/L)^\ell \le k' \le (1 + 1/L)^{\ell + 1}$, we may return $\ell$.  The strategy essentially requires to compute the powers $(1 + 1/L)^\ell$ for $\ell$'s that lie on one path of the binary search. All such powers can be computed with $O(\log L)$ arithmetic operations. \end{proof}

The computation of $x$ from $\hx$ thus requires $O(n \log(nU))$ arithmetic operations on number of bitlength 
$O(n \log(nU) + \log L)$. 

The computation of the balanced flow $\barf$ consists of two steps. First multiplying the flows on edges incident to good $c_i \in \Gamma(B(S))$ by $x_i = \barp_i/p_i$ and then recomputing the flow on the edges $(s,b_i)$. Then
augmenting the flow as in Figure~\ref{fig:aug}. Finally, augmentation to a maximum flow and balancing. 

The computation of the balanced flow is most costly. It requires $n$ max-flow computation. Each max-flow computation needs $O(n^3)$ arithmetic operations (only additions and subtractions) on numbers of bitlength $O(n \log(nU))$. 

\begin{theorem}\label{time per iteration} An iteration takes time $O(n^5 \log(nU))$ (even with the use of quadratic algorithms for multiplication and division). \end{theorem}

\subsection{Putting Everything Together}\label{sec: putting everying together}

We can now put everything together.

\begin{theorem} Assume that utilities are integers bounded by $U$. An equilibrium price vector can be computed in time $O(n^{10}\log^2(nU))$. \end{theorem}

\section{Conclusion}\label{sec: Conclusion}

We believe that the running time analysis of our algorithm can be improved. We have shown that the number of iterations is $O(n^5 \log(nU))$ and that each iteration can be implemented to run in $O(n^5 \log(nU))$. Both bounds are probably overly pessimistic. We see two possible directions for  improving the bound on the cost of an iteration: Show that a balanced flow can be computed faster than with $\Theta(n)$ maxflow computations, or show that in our algorithm not every iteration requires $n$ maxflow computations.

The development of a strongly polynomial algorithm for the linear Arrow-Debreu market is a major open problem. A strongly polynomial algorithms for the linear Fisher market was given by Orlin~\cite{Orlin10}. 

The generalization of our algorithm to more general utility functions is also interesting. For Fisher's market Vazirani~\cite{Vazirani10} extended the algorithm by Devanur et al~\cite{DPSV08} to spending constraint utilities. 


\end{document}